\documentclass[journal]{IEEEtran}

\usepackage{amsmath,amssymb,amsfonts}
\usepackage{amsthm}
\usepackage{algorithmic,algorithm}
\usepackage{graphicx}
\usepackage{booktabs,multirow}
\usepackage[table]{xcolor}
\usepackage{hyperref}
\usepackage{subcaption}
\usepackage{enumitem}
\usepackage{cite}
\usepackage{bm}
\usepackage{pifont}

\newtheorem{definition}{Definition}
\newtheorem{proposition}{Proposition}
\newtheorem{theorem}{Theorem}
\newtheorem{problem}{Problem}
\newtheorem{remark}{Remark}
\newtheorem{observation}{Observation}

\newcommand{\up}[1]{{\color{green!60!black}\scriptsize$\uparrow$#1}}
\newcommand{\down}[1]{{\color{red!70!black}\scriptsize$\downarrow$#1}}

\begin{document}

\title{\textsc{EGRefine}: An Execution-Grounded Optimization\\
Framework for Text-to-SQL Schema Refinement}

\author{Jiaqian Wang, Yutao Qi$^*$, Wenjin Hou, Yu Pang, and Rui Yang\\
Xidian University, Xi'an, Shaanxi 710071, China\\
Email: 24151111272@stu.xidian.edu.cn,\;
ytqi@xidian.edu.cn,\;
houwenjinmail@gmail.com\\
\normalfont{$^*$Corresponding author}%
}

\maketitle

\begin{abstract}
Text-to-SQL enables non-expert users to query databases in natural language,
yet real-world schemas often suffer from ambiguous, abbreviated, or inconsistent
naming conventions that degrade model accuracy.
Existing approaches treat schemas as fixed and address errors downstream.
In this paper, we frame schema refinement as a constrained optimization problem:
find a renaming function that maximizes downstream Text-to-SQL execution accuracy
while preserving query equivalence through database views.
We analyze the computational hardness of this problem,
which motivates a column-wise greedy decomposition,
and instantiate it as \textsc{EGRefine}: a four-phase
pipeline that screens ambiguous columns, generates
context-aware candidate names, verifies them through
execution-grounded feedback, and materializes the result
as non-destructive SQL views.
The pipeline carries two structural properties:
column-local non-degradation, ensured by the conservative
selection rule in the verification phase, and
database-level query equivalence, ensured by the view-based
materialization phase.
Together they make the resulting refinement safe by construction
\emph{at the column level}, with cross-column and prompt-level
interactions handled empirically rather than analytically.
Across controlled schema-degradation, real-world, and
enterprise benchmarks, \textsc{EGRefine} recovers accuracy
lost to schema naming noise where applicable and correctly
abstains where the underlying task exceeds current Text-to-SQL
capabilities, with refined schemas transferring across model
families to enable refine-once, serve-many-models deployment.
Code and data are publicly available at
\url{https://github.com/ai-jiaqian/EGRefine}.
\end{abstract}



\section{Introduction}
\label{sec:intro}

Text-to-SQL---the task of translating natural-language questions
into executable SQL queries---has emerged as a pivotal technology
for democratizing data access~\cite{katsogiannis2023survey,luo2025nl2sql}.
Driven by rapid advances in large language models (LLMs),
recent systems have achieved impressive performance on
standard benchmarks~\cite{pourreza2023dinsql,gao2023dailsql,wang2025macsql,talaei2024chess},
with execution accuracies exceeding 85\% on the widely used
Spider dataset~\cite{yu2018spider}.

However, a persistent gap remains between benchmark performance
and production reliability.
Real-world databases are frequently characterized by
inconsistent naming standards, ad-hoc abbreviations, and sparse
documentation~\cite{furst2024robustness,renggli2025challenges}.
Columns named \texttt{A2}, \texttt{nm}, \texttt{sal}, or
\texttt{dt} carry little semantic information, forcing
Text-to-SQL models to rely on fragile heuristic cues during
schema linking~\cite{wang2020ratsql}.
Diagnostic evaluations confirm that even state-of-the-art
systems suffer substantial accuracy degradation when schema
names are perturbed or
obfuscated~\cite{chang2023drspider,li2023bird},
establishing schema-level naming quality as a critical yet
underexplored bottleneck for robust Text-to-SQL.

The prevailing response has been to treat schema quality as a
fixed background condition and address its consequences
\emph{downstream}: constrained decoding improves syntactic
validity but not naming ambiguity~\cite{scholak2021picard};
error-correction and self-refinement repair generated SQL after
the fact~\cite{chen2023errorcorrection,qu2025share,mao2024executionguided},
yet schema-linking errors rooted in opaque names persist;
interactive clarification shifts the burden to
users~\cite{yu2019cosql,elgohary2020speak,tian2023interactive},
impractical at scale; and recent ambiguity benchmarks probe
query-level perturbations rather than schema-internal
defects~\cite{bhaskar2023ambiqt,saparina2024ambrosia,qiu2025practiq}.
The schema itself---the root cause---remains untreated.

A small but growing line of work has begun to consider
\emph{upstream} schema-level interventions.
Odin~\cite{ding2025odin} recommends multiple SQL
interpretations for ambiguous schemas at query time.
CLEAR~\cite{zhao2025clear} provides a parser-independent
disambiguation framework.
While these efforts represent important steps, they share
two fundamental limitations.
First, they operate at \emph{query time}, incurring per-query
overhead without producing durable schema-level improvements.
Second, they lack a \emph{grounded optimization objective}:
the quality of a schema modification is assessed by linguistic
plausibility rather than by its measurable effect on downstream
task performance.

In this paper, we propose a fundamentally different approach.
We observe that schema quality, viewed through the lens of
Text-to-SQL, can be defined objectively as the downstream
execution accuracy it induces.
This allows us to formalize schema refinement as a constrained
optimization problem: find a renaming of schema elements that
maximizes execution accuracy while preserving query equivalence
via SQL views.
We analyze the computational hardness---the search space is
exponential (Observation~\ref{prop:space}) and the constrained
problem is at least NP-hard (Theorem~\ref{thm:nphard})---and
propose \textsc{EGRefine} (Execution-Grounded Refinement),
a four-phase pipeline whose key innovation is
\emph{execution-grounded verification}: rather than relying on
an LLM's semantic judgment to select column names, we use
downstream SQL execution results as the selection signal,
treating execution accuracy as a reward function for schema
refinement.
This applies execution feedback at a novel granularity---not to
the generated query as in prior
work~\cite{chen2023codet,chen2024selfdebugging,shinn2023reflexion},
but to the schema itself.

Our contributions are as follows:

\begin{enumerate}[leftmargin=2em]
  \item \textbf{Formal Problem Framing.}\;
    We cast schema refinement for Text-to-SQL as a constrained
    optimization problem with a task-grounded objective
    (Definition~\ref{def:quality}), analyze its computational
    complexity (\S\ref{ssec:hardness}), and establish two
    structural properties: \emph{column-local non-degradation}
    (each isolated renaming is constrained to be non-harmful on
    its query subset, Proposition~\ref{prop:mono}) and
    \emph{query equivalence at the database level} under
    standard DBMS view-expansion semantics
    (Proposition~\ref{thm:equiv}).

  \item \textbf{The \textsc{EGRefine} Framework.}\;
    We propose a four-phase pipeline operationalizing the
    optimization framework: heuristic pruning reduces the search
    space (\S\ref{ssec:phase1}), LLM generation proposes candidate
    renamings (\S\ref{ssec:phase2}), execution-grounded verification
    selects the optimal candidate via downstream SQL execution
    (\S\ref{ssec:phase3}), and VIEW synthesis materializes the
    result as a non-destructive semantic layer
    (\S\ref{ssec:phase4}).
    The pipeline is model-agnostic: any downstream Text-to-SQL
    system benefits from the refined schema without modification.

  \item \textbf{Comprehensive Empirical Evaluation.}\;
    We evaluate on three benchmarks playing complementary evidentiary roles:
    Dr.Spider~\cite{chang2023drspider} establishes the mechanism under
    controlled degradation;
    BIRD~\cite{li2023bird} shows smaller but measurable gains on natural
    schemas, with workload-holdout (\S\ref{ssec:rq_holdout}) confirming
    persistence on unseen queries;
    BEAVER~\cite{chen2024beaver} delineates the applicability boundary,
    where SQL-generation rather than schema-linking is the bottleneck and
    Phase~3 correctly abstains.
    Regressions on well-named schemas are rare and small under the
    conservative rule, ablations confirm execution-grounded verification
    significantly outperforms pure LLM-based selection, and benefits amplify
    for weaker models---the setting most relevant to cost-sensitive deployments.
    Our mechanism analysis (\S\ref{sec:analysis}) additionally yields
    a coverage--improvement linearity that makes expected impact predictable
    at deployment time.
\end{enumerate}


\section{Related Work}
\label{sec:related}

We review three lines of research relevant to our work:
Text-to-SQL systems, schema ambiguity and robustness, and
execution-guided methods.

\subsection{Text-to-SQL Systems}
\label{ssec:rw_text2sql}

Text-to-SQL has progressed through several paradigm shifts:
early sequence-to-sequence
generation~\cite{zhong2017seq2sql} and schema-aware
encoders~\cite{wang2020ratsql,lin2020bridge}; constrained decoding for
syntactic validity~\cite{scholak2021picard}; and decoupled
schema-linking-then-parsing
architectures~\cite{li2023resdsql}.
The advent of LLMs shifted the dominant paradigm toward
prompt-based methods such as
DIN-SQL~\cite{pourreza2023dinsql} (decomposed prompting),
DAIL-SQL~\cite{gao2023dailsql} (token-efficient skeleton-based
prompts), and multi-agent systems including
MAC-SQL~\cite{wang2025macsql} and
CHESS~\cite{talaei2024chess}; fine-tuned variants such as
CodeS~\cite{li2024codes} also appear in this landscape.
Recent work~\cite{maamari2024deathsl} questions whether
explicit schema linking is needed at all when LLMs have
sufficient context; in either regime, schema-side quality
remains a precondition for accurate generation.
Comprehensive surveys
include~\cite{katsogiannis2023survey,xinyu2025survey,hong2025survey,luo2025nl2sql}.

All of the above treat the database schema as a fixed input.
\textsc{EGRefine} is orthogonal: it operates as an upstream
preprocessing layer that refines the schema before any
Text-to-SQL system is invoked, composable with any of the
above without modification.

\subsection{Schema Ambiguity and Robustness}
\label{ssec:rw_ambiguity}

\subsubsection{Robustness Evaluation}

Earlier benchmarks established schema-side brittleness: 
Spider-Syn~\cite{gan2021spidersyn} showed accuracy collapses 
under synonym substitution that breaks lexical matching 
between question tokens and schema names, and 
Spider-DK~\cite{gan2021spiderdk} extended this to domain-knowledge 
generalization.
Dr.Spider~\cite{chang2023drspider} subsequently introduced 17 perturbation
types---including schema-level modifications---and showed
substantial accuracy drops even for the strongest models.
F\"{u}rst et al.~\cite{furst2024robustness} evaluated robustness
to data model variations using real user queries, while
Renggli et al.~\cite{renggli2025challenges} identified
fundamental challenges in evaluating Text-to-SQL under
realistic conditions.
These studies diagnose rather than resolve schema-side brittleness.

\subsubsection{Query-Level Ambiguity}

A separate line of work addresses ambiguity in the
\emph{question} rather than the
\emph{schema}~\cite{bhaskar2023ambiqt,saparina2024ambrosia,qiu2025practiq},
including interactive disambiguation that solicits user
clarification~\cite{elgohary2020speak,tian2023interactive,qiu2025disambiguation}
and detection of unanswerable
questions~\cite{wang2023ambiguity}.
These are complementary: we address ambiguity in the schema,
not the question.

\subsubsection{Schema-Level Interventions}

Several benchmarks acknowledge schema quality as a real-world
challenge: BIRD~\cite{li2023bird} provides supplementary
knowledge descriptions alongside schemas;
Spider 2.0~\cite{lei2024spider2} features enterprise-scale
schemas with over 1{,}000 columns and domain-specific
abbreviations;
BEAVER~\cite{chen2024beaver} provides enterprise data
warehouse schemas with extremely low Text-to-SQL baselines;
and BenchPress~\cite{wenz2026benchpress} addresses enterprise
schema ambiguity through human-in-the-loop annotation.

Two recent systems directly tackle schema-level disambiguation.
Odin~\cite{ding2025odin} recommends multiple SQL candidates
based on alternative schema interpretations and learns from
user feedback.
CLEAR~\cite{zhao2025clear} provides a parser-independent
framework that generates and selects among candidate SQL
interpretations.
Chen et al.~\cite{chen2025reliable} address schema linking
uncertainty through adaptive abstention.
Closer to our setting, SNAILS~\cite{luoma2025snails} shows that
schema-identifier naturalness materially affects LLM-based
NL-to-SQL accuracy and proposes natural views as a mitigation;
\textsc{EGRefine} extends this direction by tying view
construction to an execution-grounded optimization objective
rather than a naturalness classifier.
These methods operate at \emph{query time}---each incoming
question triggers disambiguation---and evaluate schema
modifications through linguistic plausibility or user feedback.
\textsc{EGRefine} differs along two axes: it performs schema
refinement \emph{once} as offline preprocessing, producing a
durable refined schema; and it selects refinements through a
grounded optimization objective tied directly to downstream
execution accuracy.

\subsection{Execution-Guided Methods}
\label{ssec:rw_execution}

In Text-to-SQL, execution feedback has been used to constrain
or re-rank candidate
queries~\cite{scholak2021picard}, iteratively repair
SQL~\cite{mao2024executionguided,chen2023errorcorrection},
and drive self-correction loops in multi-agent
frameworks~\cite{wang2025macsql,chaturvedi2025sqlofthought}.
Beyond SQL, execution-based verification is effective in
general code generation, including dual-execution
agreement~\cite{chen2023codet}, self-debugging from
execution traces~\cite{chen2024selfdebugging}, and verbal
reinforcement learning over execution
feedback~\cite{shinn2023reflexion}.

Our work draws inspiration from this paradigm but applies
execution feedback at a different granularity.
Prior methods use execution signals to improve the
\emph{generated output} while holding the schema constant;
\textsc{EGRefine} uses them to improve the \emph{input
representation} itself, producing a durable refined schema
that benefits all subsequent queries without per-query loops.
Table~\ref{tab:comparison} summarizes the distinctions.

\begin{table}[t]
\centering
\caption{Comparison with closely related approaches.}
\label{tab:comparison}
\renewcommand{\arraystretch}{1.15}
\setlength{\tabcolsep}{3pt}
\begin{tabular}{lccccc}
\toprule
\textbf{Method} &
\rotatebox{65}{\textbf{Upstream}} &
\rotatebox{65}{\textbf{Offline}} &
\rotatebox{65}{\textbf{Exec.-grounded}} &
\rotatebox{65}{\textbf{Model-agnostic}} &
\rotatebox{65}{\textbf{Non-destructive}} \\
\midrule
PICARD~\cite{scholak2021picard}          & \ding{55} & \ding{55} & \ding{55} & \ding{55} & -- \\
DART-SQL~\cite{mao2024executionguided}   & \ding{55} & \ding{55} & \ding{51} & \ding{55} & -- \\
Odin~\cite{ding2025odin}                 & \ding{51} & \ding{55} & \ding{55} & \ding{51} & \ding{51} \\
CLEAR~\cite{zhao2025clear}               & \ding{51} & \ding{55} & \ding{55} & \ding{51} & \ding{51} \\
\textbf{EGRefine (Ours)}                 & \ding{51} & \ding{51} & \ding{51} & \ding{51} & \ding{51} \\
\bottomrule
\end{tabular}
\end{table}


\section{Problem Formulation}
\label{sec:formulation}

In this section, we formalize schema refinement as a constrained
optimization problem, analyze its computational hardness, and
present our decomposition strategy with associated structural
properties.

\subsection{Preliminaries and Notation}
\label{ssec:prelim}

\textbf{Database Schema.}
A relational database is described by a schema
$\mathcal{S} = (\mathcal{T}, \mathcal{C}, \mathcal{F})$,
where $\mathcal{T} = \{t_1, \dots, t_n\}$ is the set of tables,
$\mathcal{C} = \{c_1, \dots, c_m\}$ is the set of columns,
and $\mathcal{F} \subseteq \mathcal{C} \times \mathcal{C}$ encodes foreign-key relationships.
Each column $c \in \mathcal{C}$ belongs to a unique table $\mathrm{tab}(c) \in \mathcal{T}$
and carries a \emph{surface name} $\mathrm{name}(c)$---the identifier visible to downstream systems.
We define the \emph{scope} of a column $c$ as the set of columns that share
its naming namespace:
\begin{multline}
  \mathrm{scope}(c) = \{c' \in \mathcal{C} \mid \mathrm{tab}(c') = \mathrm{tab}(c)\} \\
  \cup\; \{c' \mid (c, c') \in \mathcal{F} \text{ or } (c', c) \in \mathcal{F}\}.
  \label{eq:scope}
\end{multline}
Scope includes same-table columns (where SQL forbids duplicate
column names) and FK-related columns across tables (where duplicates
create joining ambiguity).
Cross-table homonyms without FK relationships are permitted by
SQL and left untouched (\S\ref{ssec:phase1}).

\textbf{Text-to-SQL Task.}
A Text-to-SQL model $M$ takes a natural-language question $\mathit{nl}$
and a schema $\mathcal{S}$ as input, and produces a predicted SQL query
$M(\mathit{nl}, \mathcal{S})$.
Performance is measured by \emph{execution accuracy} (ExAcc):
given a benchmark $Q = \{(\mathit{nl}_j, \mathit{sql}_j^*)\}_{j=1}^{|Q|}$
of question--gold-SQL pairs, we define
\begin{multline}
  \mathrm{ExAcc}(M, \mathcal{S}, Q)
  = \frac{1}{|Q|} \sum_{j=1}^{|Q|}
    \mathbb{1}\bigl[
      \mathrm{exec}(M(\mathit{nl}_j, \mathcal{S})) \\
      = \mathrm{exec}(\mathit{sql}_j^*)
    \bigr],
  \label{eq:exacc}
\end{multline}
where $\mathrm{exec}(\cdot)$ denotes the result set obtained by executing
a query on the underlying database.

\subsection{Schema Refinement as Optimization}
\label{ssec:objective}

We adopt a \emph{task-grounded} perspective: schema quality is
determined solely by its effect on downstream Text-to-SQL
performance, avoiding subjective ambiguity taxonomies.

\begin{definition}[Refinement Mapping]
\label{def:refinement}
A \emph{refinement mapping} $r: \mathcal{C} \to \Sigma^*$ assigns
a (possibly new) surface name $r(c)$ to each column $c \in \mathcal{C}$,
where $\Sigma^*$ is the set of valid SQL identifiers.
The identity mapping $r_\mathrm{id}$ corresponds to the unchanged
schema; we write $r(\mathcal{S})$ for the schema obtained by
applying $r$.
A refinement is \emph{admissible} if (i) no two columns in the
same scope receive identical names, and (ii) only surface names
are modified (data types, constraints, foreign keys, and table
structure remain unchanged).
We denote the set of admissible refinements by $\mathcal{R}(\mathcal{S})$.
\end{definition}

\begin{definition}[Schema Quality]
\label{def:quality}
Given $\mathcal{S}$, a set of Text-to-SQL models
$\mathcal{M} = \{M_1, \dots, M_l\}$, and a query set $Q$:
\begin{equation}
  \mathrm{Quality}(\mathcal{S},\, \mathcal{M},\, Q)
  = \frac{1}{|\mathcal{M}|}
    \sum_{M_j \in \mathcal{M}} \mathrm{ExAcc}(M_j,\, \mathcal{S},\, Q).
  \label{eq:quality}
\end{equation}
Averaging over multiple models ensures the metric reflects
general schema clarity rather than single-system idiosyncrasies.
\end{definition}

\begin{problem}[Optimal Schema Refinement]
\label{prob:main}
Given $\mathcal{S}$, $\mathcal{M}$, and $Q$, find:
\begin{equation}
  r^* = \arg\max_{r \,\in\, \mathcal{R}(\mathcal{S})}\;
         \mathrm{Quality}\bigl(r(\mathcal{S}),\, \mathcal{M},\, Q\bigr),
  \label{eq:optim}
\end{equation}
subject to: for every SQL query $q'$ over $r(\mathcal{S})$,
$\mathrm{exec}(q',\,r(\mathcal{S})) = \mathrm{exec}(q,\,\mathcal{S})$
for the semantically equivalent query $q$ over $\mathcal{S}$.
\end{problem}

The equivalence constraint ensures refinement is
non-destructive: the original database remains untouched.
Proposition~\ref{thm:equiv} shows this is satisfied by
construction when the refined schema is materialized as SQL
views over the original tables.

\begin{remark}
A column $c$ is considered \emph{refinement-improvable} iff
there exists an alternative name $c'$ such that replacing
$\mathrm{name}(c)$ with $c'$ strictly increases
$\mathrm{Quality}$.
This operational definition sidesteps ambiguity taxonomies
and ties refinement directly to measurable task improvement.
\end{remark}

\subsection{Hardness Analysis}
\label{ssec:hardness}

\begin{observation}[Exponential Search Space]
\label{prop:space}
Let $|\mathcal{C}| = m$ and suppose each column has $k$ candidates
(including the original).
Then $|\mathcal{R}(\mathcal{S})| = O(k^m)$: for $m{=}100$ and
$k{=}3$, $|\mathcal{R}(\mathcal{S})| > 5 \times 10^{47}$,
rendering exhaustive search infeasible.
\end{observation}

To establish hardness rigorously, we introduce a constrained
decision variant isolating the combinatorial core.

\begin{definition}[Constrained Refinement Decision Problem]
\label{def:cr_decide}
\textsc{Constrained-Refinement-Decision (CRD).}
Given a schema $\mathcal{S}$ with column set
$\mathcal{C} = \{c_1,\dots,c_n\}$, within-scope conflict relation
$E \subseteq \mathcal{C} \times \mathcal{C}$, candidate lists
$\mathcal{L}(c_i) \subseteq \Sigma^{+}$
($|\mathcal{L}(c_i)| \geq 1$), and a forced-rename subset
$\mathcal{C}^{\dagger} \subseteq \mathcal{C}$,
does there exist an assignment $r : \mathcal{C} \to \Sigma^{+}$
satisfying (i) $r(c_i) \in \mathcal{L}(c_i)$ for all $i$;
(ii) $r(c_i) \neq r(c_j)$ for all $(c_i, c_j) \in E$; and
(iii) $r(c_i) \neq c_i^{0}$ for all
$c_i \in \mathcal{C}^{\dagger}$?
\end{definition}

\begin{theorem}[NP-hardness of Constrained Refinement]
\label{thm:nphard}
\textsc{Constrained-Refinement-Decision} is NP-hard.
The optimization version of Problem~\ref{prob:main}---maximizing
$\mathrm{Quality}$ over assignments satisfying the same
constraints---is therefore also NP-hard.
\end{theorem}

\begin{IEEEproof}
We reduce from \textsc{List-Coloring}, NP-complete for general
graphs~\cite{jansen1997complexity}.
A \textsc{List-Coloring} instance is a graph $G = (V, E_G)$ with
color lists $L(v) \subseteq \mathbb{N}$, asking whether a proper
coloring $\phi$ exists with $\phi(v) \in L(v)$ and
$\phi(u) \neq \phi(v)$ for all $(u, v) \in E_G$.

Given $(G, L)$, construct a CRD instance: for each $v_i \in V$,
introduce a column $c_i$ under a single shared scope; set
$E = E_G$ and $\mathcal{L}(c_i) = \{\sigma(\ell) : \ell \in L(v_i)\}$
for an injective encoding $\sigma : \mathbb{N} \to \Sigma^{+}$;
set $\mathcal{C}^{\dagger} = \emptyset$.
The construction is polynomial.
A valid coloring $\phi$ yields a valid CRD assignment
$r(c_i) = \sigma(\phi(v_i))$ (condition (iii) is vacuous since
$\mathcal{C}^{\dagger} = \emptyset$); conversely, a valid CRD
assignment $r$ yields a valid coloring
$\phi(v_i) = \sigma^{-1}(r(c_i))$ via injectivity of $\sigma$.
The optimization version of Problem~\ref{prob:main} contains CRD
as a feasibility sub-problem, so it is at least as hard.
\end{IEEEproof}

\begin{remark}[Scope of the Hardness Result]
\label{rem:hardness_scope}
Theorem~\ref{thm:nphard} captures hardness from the combinatorial
structure of constrained candidate assignment, not from properties
of $\mathrm{Quality}$.
Inapproximability results exploiting structure of $\mathrm{Quality}$
itself remain open.
\end{remark}

\subsection{Decomposition and Structural Properties}
\label{ssec:decomposition}

Since exact optimization is intractable
(Observation~\ref{prop:space}, Theorem~\ref{thm:nphard}), we
decompose the joint problem into independent per-element
subproblems, analogous to coordinate descent.

\subsubsection{Element-Wise Refinement}

For each column $c_i \in \mathcal{C}$, let
$Q(c_i) = \{(\mathit{nl}, \mathit{sql}^*) \in Q
          \mid c_i \text{ appears in } \mathit{sql}^*\}$
denote the queries referencing $c_i$.
We solve the following subproblem independently for each $c_i$
identified as a refinement candidate (\S\ref{sec:method}):
\begin{equation}
  c_i^* = \arg\max_{c' \,\in\, \mathrm{Cand}(c_i) \,\cup\, \{c_i^0\}}
           \mathrm{Quality}\bigl(\mathcal{S}[c_i \!\to\! c'],\,
                                 \mathcal{M},\, Q(c_i)\bigr),
  \label{eq:elementwise}
\end{equation}
where $c_i^0 = \mathrm{name}(c_i)$ is the original name,
$\mathrm{Cand}(c_i)$ is the set of $k$ candidates from the
proposal module, and $\mathcal{S}[c_i \!\to\! c']$ denotes
$\mathcal{S}$ with the substitution
$\mathrm{name}(c_i) \gets c'$.

\subsubsection{Conservative Selection Rule}

To prevent regressions, we adopt a conservative policy: a column
is renamed only if the best candidate strictly improves on the
original.
Let
\begin{multline}
\Delta_i = \max_{c' \in \mathrm{Cand}(c_i)} 
\mathrm{Quality}(\mathcal{S}[c_i \!\to\! c'], \mathcal{M}, Q(c_i)) \\
- \mathrm{Quality}(\mathcal{S}, \mathcal{M}, Q(c_i))
\label{eq:delta_i}
\end{multline}
denote the quality gain of the best candidate over the original.
The conservative selection is:
\begin{equation}
  c_i^* =
  \begin{cases}
    \arg\max_{c'} \mathrm{Quality}(\cdot),
    & \text{if } \Delta_i > 0; \\
    c_i^0, & \text{otherwise.}
  \end{cases}
  \label{eq:conservative}
\end{equation}

This rule yields a local guarantee at the per-column level.

\begin{proposition}[Per-Column Local Non-Degradation]
\label{prop:mono}
Consider the conservative selection rule~\eqref{eq:conservative}
applied to a column $c_i$, holding all other columns fixed.
Let $\mathcal{S}_{c_i \to c_i^0}$ denote the schema with $c_i$
keeping its original name, and $\mathcal{S}_{c_i \to c_i^*}$
the schema after applying the rule's selected name.
Then on the column-local query subset $Q(c_i)$:
\begin{equation}
  \mathrm{Quality}(\mathcal{S}_{c_i \to c_i^*},\mathcal{M},Q(c_i))
  \;\geq\;
  \mathrm{Quality}(\mathcal{S}_{c_i \to c_i^0},\mathcal{M},Q(c_i)).
  \label{eq:mono}
\end{equation}
\end{proposition}

\begin{IEEEproof}
By the rule~\eqref{eq:conservative}, the candidate $c_i^*$ is
adopted only if $\Delta_i > 0$, where $\Delta_i$ is defined
as the change in $\mathrm{Quality}$ on $Q(c_i)$ between
$\mathcal{S}_{c_i \to c_i^*}$ and $\mathcal{S}_{c_i \to c_i^0}$.
Otherwise, the original name is retained and the two schemas are
identical on $Q(c_i)$.
Either way, inequality~\eqref{eq:mono} holds.
\end{IEEEproof}

\begin{remark}[Scope of the Guarantee]
\label{rem:mono_scope}
Proposition~\ref{prop:mono} is a \emph{per-column local}
property: it constrains only the change on $Q(c_i)$ when column
$c_i$ alone is renamed.
It does not extend to a global non-degradation guarantee on the
full workload, for two reasons.
First, LLM-based Text-to-SQL systems consume the entire schema
as prompt context, so renaming one column can affect behavior on
queries that do not reference $c_i$---through schema-linking
attention, candidate-table pruning, or query-decomposition
strategy---effects not bounded by the column-local $\Delta_i$.
Second, simultaneously applying renamings to multiple columns
can produce schema-level interactions that single-column
verification does not detect, since Phase~3 evaluates each
candidate in isolation.
We therefore characterize \textsc{EGRefine} as
\emph{empirically regression-resistant} rather than provably
non-degrading: the conservative rule substantially reduces the
risk of harmful renamings (15/18 configurations show net
positive $\Delta$ in \S\ref{ssec:flip_analysis}) but does not
eliminate it under cross-column or prompt-level interactions.
\end{remark}

\subsubsection{Global Conflict Resolution}
\label{ssec:conflict}

After all per-element selections are made, we perform a global
consistency check.
If two columns $c_i, c_j$ within the same scope have been assigned
identical refined names ($c_i^* = c_j^*$), a \emph{naming collision}
arises, violating the admissibility constraint.
We resolve such collisions by retaining the refinement with the higher
$\Delta$ and reverting the other to its next-best non-conflicting candidate.
This process repeats for at most two iterations, after which any
remaining conflicts are resolved by retaining the original names.

\subsubsection{Query Equivalence via VIEW Synthesis}

The refined schema $r(\mathcal{S})$ is materialized as a set of
SQL \texttt{CREATE VIEW} statements aliasing original column
names to their refined counterparts:
\begin{multline}
  \texttt{CREATE VIEW } \hat{t} \texttt{ AS SELECT} \\
  c_1 \texttt{ AS } r(c_1),\; \dots,\;
  c_n \texttt{ AS } r(c_n) \texttt{ FROM } t,
  \label{eq:view}
\end{multline}
for each table $t$ with columns $c_1, \dots, c_n$.
Downstream models generate SQL referencing the refined names;
the database engine resolves these to underlying columns during
query planning, with no intermediate string rewriting.

\begin{proposition}[Query Equivalence under DBMS View Semantics]
\label{thm:equiv}
Let $r$ be an admissible refinement mapping, and let
$V = \mathrm{VIEW}(r, \mathcal{S})$ be the corresponding view
definitions.
For any read-only SQL query $q'$ over $r(\mathcal{S})$
consisting of relational-algebra core operators (projection,
selection, equi-join, set operations) over the views, there
exists a semantically equivalent query $q$ over $\mathcal{S}$
such that:
\begin{equation}
  \mathrm{exec}(q,\, \mathcal{S}) = \mathrm{exec}(q',\, r(\mathcal{S})).
  \label{eq:equiv}
\end{equation}
For SQL features beyond the core algebra (aggregation,
\texttt{ORDER BY}, \texttt{LIMIT}, \texttt{DISTINCT}, scalar
functions, quoted identifiers, \texttt{NATURAL JOIN}),
equivalence is provided by DBMS view-expansion
semantics~\cite{garcia2008database}.
\end{proposition}

\begin{IEEEproof}
For the relational-algebra core, view
definition~\eqref{eq:view} implements a rename operation
$\rho$: $\hat{t} = \rho_{r}(t)$.
By the commutativity of $\rho$ with selection, projection,
equi-join, and set operations~\cite{abiteboul1995foundations},
and admissibility of $r$ (so $r^{-1}$ is well-defined within
each scope), any expression $e'$ over $r(\mathcal{S})$ admits
an equivalent rewrite $e = e'[r^{-1}]$ over $\mathcal{S}$,
establishing~\eqref{eq:equiv}.
For SQL constructs outside this core, equivalence is delegated
to DBMS view-expansion semantics: a query against a view is
evaluated as if the view definition were textually
substituted~\cite{garcia2008database}.
We verified empirically across all benchmarks (Dr.Spider, BIRD,
BEAVER) that no execution discrepancy arises between gold SQL
run against $\mathcal{S}$ and its $r$-aliased counterpart run
against $V$.
\end{IEEEproof}

\medskip

Together, Observation~\ref{prop:space},
Theorem~\ref{thm:nphard}, Proposition~\ref{prop:mono} (with
Remark~\ref{rem:mono_scope}), and Proposition~\ref{thm:equiv}
provide the formal underpinning for our method: the problem is
intractable in general, but the greedy decomposition is locally
non-degrading at the column level and the output is provably
non-destructive at the database level.
Section~\ref{sec:method} presents \textsc{EGRefine}, a concrete
instantiation of this framework.

\begin{figure*}[t]
\centering
\includegraphics[width=0.96\textwidth]{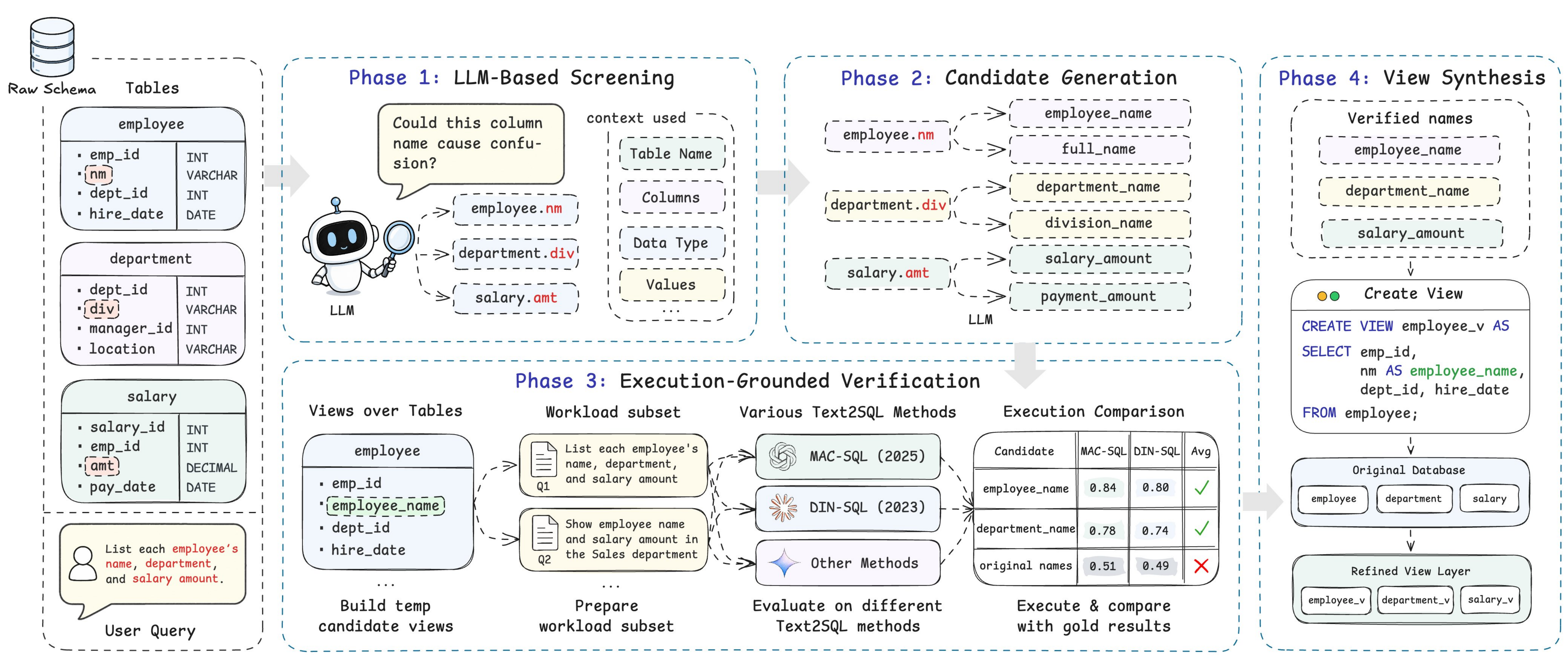}
\caption{Overview of the \textsc{EGRefine} pipeline.
Given a raw schema $\mathcal{S}$ (left, with example tables
\texttt{employee}, \texttt{department}, \texttt{salary}),
the four phases progressively refine column names and output
a refined schema $\mathcal{S}'$ as a non-destructive view layer
(right).
\textbf{Phase~1} (LLM-Based Screening, \S\ref{ssec:phase1})
selects $n \ll m$ candidate columns whose names may impede
Text-to-SQL interpretation, using full-schema LLM judgment
rather than surface-level rules.
\textbf{Phase~2} (Candidate Generation, \S\ref{ssec:phase2})
proposes $k$ context-aware renaming candidates per selected
column, conditioned on schema context, neighbor column values,
and conservative renaming guidelines.
\textbf{Phase~3} (Execution-Grounded Verification,
\S\ref{ssec:phase3}) builds a temporary view per candidate,
prepares a workload subset, runs multiple Text-to-SQL methods
(e.g., MAC-SQL, DIN-SQL), and compares execution against gold
results to select $c^* = \arg\max_{c_j} \Delta(c_j)$;
the conservative rule keeps the original name when
$\Delta < \tau_{\min}$.
\textbf{Phase~4} (VIEW Synthesis, \S\ref{ssec:phase4})
materializes accepted refinements as SQL views, with query
equivalence delegated to engine-level view expansion
(\S\ref{ssec:decomposition}).
The original database remains untouched.
\textit{By design, the user query workload (bottom-left) is
consumed only by Phase~3; Phases 1, 2, and 4 are schema-only,
so candidate generation does not overfit to current workload
phrasing and supervision is concentrated at a single,
replaceable bottleneck.}}
\label{fig:pipeline}
\end{figure*}


\section{Method: \textsc{EGRefine}}
\label{sec:method}

\textsc{EGRefine} is a four-phase pipeline that approximately
solves Problem~\ref{prob:main} via the column-wise decomposition
of \S\ref{sec:formulation}.
Algorithm~\ref{alg:pipeline} provides a high-level overview.

\noindent\textbf{Deployment setting.}\;
\textsc{EGRefine} is an offline preprocessing step: it produces
a refined schema once and amortizes the cost over future queries.
The input includes a representative query workload $Q$ with
ground-truth SQL---reflecting deployments where business
questions and expert-authored SQL characterize how the schema
is used.
We use benchmark dev sets as proxies; settings without
ground-truth SQL are outside scope (\S\ref{sec:conclusion}).

\begin{algorithm}[t]
\caption{\textsc{EGRefine} Pipeline}
\label{alg:pipeline}
\begin{algorithmic}[1]
\REQUIRE Schema $\mathcal{S}$, models $\mathcal{M}$,
  queries $Q$, data samples $D$
\ENSURE Refined schema $\mathcal{S}^*$ (as SQL VIEWs)
\STATE $A \gets \textsc{Phase1-Screen}(\mathcal{S}, \mathcal{M}_{\mathrm{screen}})$
  \hfill $\triangleright$ LLM screening
\STATE $A \gets A \setminus \textsc{StructuralExclude}(\mathcal{S})$
  \hfill $\triangleright$ Remove FK / join keys
\FOR{each column $c_i \in A$}
  \STATE $\mathrm{Cand}(c_i) \gets \textsc{Phase2-Generate}(c_i, \mathcal{S}, D)$
  \STATE $c_i^*, \Delta_i \gets \textsc{Phase3-Verify}(c_i, \mathrm{Cand}(c_i), Q, \mathcal{M})$
\ENDFOR
\STATE $R \gets \{(c_i, c_i^*) \mid \Delta_i \geq \tau_{\min}\}$
\STATE $R \gets R \cup \textsc{PropagatePK}(R, \mathcal{S})$
  \hfill $\triangleright$ PK$\to$FK
\STATE $\mathcal{S}^* \gets \textsc{Phase4-Synthesize}(R, \mathcal{S})$
  \hfill $\triangleright$ CREATE VIEW
\RETURN $\mathcal{S}^*$
\end{algorithmic}
\end{algorithm}

\subsection{Phase 1: LLM-Based Schema Screening}
\label{ssec:phase1}

The refinement search space is exponential
(Observation~\ref{prop:space}), so Phase~1 selects a subset
$A \subseteq \mathcal{C}$ of refinement-improvable columns
($|A| = n \ll m$).
We use a lightweight LLM (Qwen3.5-27B) as a schema quality
assessor, prompting it with the binary question \emph{``Could
this column name cause confusion for a Text-to-SQL model?''} and
five contextual signals: database domain description, table name
and column data type, neighboring column names, 5 sample rows,
and explicit criteria covering abbreviations, domain-specific
polysemy, single-letter codes, and generic vocabulary.
LLM screening captures semantic ambiguity beyond surface-level
rules---e.g., \texttt{label} passes lexical checks yet denotes a
carcinogenicity indicator in toxicology.

\noindent\textbf{Structural Exclusion.}\;
A structural filter preserves join integrity: primary keys are
allowed (with refined names propagated to FKs in Phase~4);
foreign keys are excluded (driven by the corresponding PK);
cross-table homonyms sharing both name and data type are excluded
as likely implicit join keys.
This design favors recall---false positives incur no harm since
Phase~3's conservative rule retains their original names.

\subsection{Phase 2: Context-Aware Candidate Generation}
\label{ssec:phase2}

For each candidate column $c_i \in A$, we prompt an LLM to
generate $k$ alternative names (default $k{=}3$) using four
context elements: a one-sentence database domain description;
neighboring columns with names, types, and 5 sample values;
20 sample values from the target column; and conservative
renaming guidelines (include the original name if already
clear, avoid over-specification, interpret data values in light
of the domain).
The LLM returns a ranked JSON list, forming the augmented set
$\mathrm{Cand}^+(c_i) = \{c_i^0\} \cup \mathrm{Cand}(c_i)$.

\subsection{Phase 3: Execution-Grounded Verification}
\label{ssec:phase3}

Instead of relying on the LLM's subjective ranking, Phase~3
evaluates each candidate by its downstream effect on Text-to-SQL
execution accuracy. This is conceptually related to
self-consistency~\cite{wang2023selfconsistency}, which selects
the most consistent answer across multiple sampled outputs;
we replace the abstract notion of consistency with the concrete
signal of execution accuracy on a verifier query subset, and
apply it at the schema-element rather than per-query
granularity.

\subsubsection{Query Subset Selection}

For each candidate column $c_i$, we extract the relevant query
subset $Q(c_i) = \{(nl,\, sql^*) \in Q \mid c_i \in \mathrm{cols}(sql^*)\}$.
If $Q(c_i) = \emptyset$, execution-based verification is
impossible and we retain the original name.

\subsubsection{Per-Candidate Scoring}

For each candidate $c_j' \in \mathrm{Cand}^+(c_i)$, we
(1) construct a temporary schema $\mathcal{S}[c_i \to c_j']$ via
\texttt{CREATE VIEW}; (2) for each $nl \in Q(c_i)$, invoke the
Text-to-SQL model on the modified schema; (3) execute predicted
SQL against the views and compare with gold SQL on $\mathcal{S}$
(equivalence by Proposition~\ref{thm:equiv}).
The resulting execution accuracy is:
\begin{equation}
  \mathrm{Score}(c_j') = \mathrm{ExAcc}(M,\,
    \mathcal{S}[c_i \to c_j'],\, Q(c_i)).
  \label{eq:score}
\end{equation}

\subsubsection{Multi-Algorithm Aggregation}

To reduce dependence on any single Text-to-SQL system, we
aggregate scores:
\begin{equation}
  \overline{\mathrm{Score}}(c_j') = \frac{1}{|\mathcal{M}|}
    \sum_{M \in \mathcal{M}} \mathrm{Score}_M(c_j').
\end{equation}
We use $\mathcal{M} = \{\text{C3}, \text{DIN-SQL}\}$, holding out
MAC-SQL to test cross-algorithm transfer (\S\ref{ssec:rq3}).
$|\mathcal{M}|{=}2$ is empirically motivated
(\S\ref{ssec:rq2}): $M{=}1$ overfits and degrades other
algorithms, while $M{=}2$ forces consensus that generalizes.

\subsubsection{Conservative Selection}

The best candidate is selected as:
\begin{equation}
  c_i^* = \arg\max_{c_j' \in \mathrm{Cand}^+(c_i)}
    \overline{\mathrm{Score}}(c_j'),
  \label{eq:selection}
\end{equation}
subject to:
\begin{equation}
  \Delta_i = \overline{\mathrm{Score}}(c_i^*) -
    \overline{\mathrm{Score}}(c_i^0) \geq \tau_{\min},
  \label{eq:conservative_tau}
\end{equation}
where $\tau_{\min}$ is a minimum improvement threshold; if
$\Delta_i < \tau_{\min}$, the original name is retained.
This filters marginal noise and instantiates
Proposition~\ref{prop:mono}; Algorithm~\ref{alg:verify} details
the procedure.

\begin{algorithm}[t]
\caption{Execution-Grounded Verification (Phase 3)}
\label{alg:verify}
\begin{algorithmic}[1]
\REQUIRE Column $c_i$, candidates $\mathrm{Cand}^+(c_i)$,
  queries $Q(c_i)$, models $\mathcal{M}$, data samples $D_s$
\ENSURE Selected name $c_i^*$, improvement $\Delta_i$
\IF{$Q(c_i) = \emptyset$}
  \RETURN $c_i^0$, $\mathit{skipped}$
    \hfill $\triangleright$ retain original: no queries to verify
\ENDIF
\FOR{each $c_j' \in \mathrm{Cand}^+(c_i)$}
  \STATE $\overline{\mathrm{Score}}(c_j') \gets 0$
  \FOR{each $M \in \mathcal{M}$}
    \FOR{each $(nl, sql^*) \in Q(c_i)$}
      \STATE $\hat{sql} \gets M(nl, \mathcal{S}[c_i \to c_j'], D_s)$
      \STATE Execute $\hat{sql}$ on view $\mathcal{S}[c_i \to c_j']$;
        compare with $\mathrm{exec}(sql^*)$ on $\mathcal{S}$
    \ENDFOR
    \STATE $\overline{\mathrm{Score}}(c_j') \mathrel{+}= \mathrm{ExAcc} / |\mathcal{M}|$
  \ENDFOR
\ENDFOR
\STATE $c_i^* \gets \arg\max_{c_j'} \overline{\mathrm{Score}}(c_j')$
\STATE $\Delta_i \gets \overline{\mathrm{Score}}(c_i^*) - \overline{\mathrm{Score}}(c_i^0)$
\IF{$\Delta_i < \tau_{\min}$}
  \RETURN $c_i^0$, $0$
    \hfill $\triangleright$ conservative: keep original
\ENDIF
\RETURN $c_i^*$, $\Delta_i$
\end{algorithmic}
\end{algorithm}

\subsubsection{Computational Cost}
\label{ssec:phase3_cost}

Phase~3 dominates the offline cost:
$O\bigl(|A| \cdot k \cdot |\mathcal{M}| \cdot \overline{|Q(c)|}\bigr)$
Text-to-SQL inferences.
On Dr.Spider-Abbr (Qwen3.5-27B) with $|A|{=}665$, $k{=}3$,
$|\mathcal{M}|{=}2$, $\overline{|Q(c)|}{\approx}10$, this yields
$\approx 4 \times 10^4$ inferences ($\approx$4--6\,h on one
A800 at 16-parallel concurrency); on BIRD,
$\approx 9 \times 10^3$ inferences ($\approx$1--2\,h).
This one-time cost amortizes across future queries and multiple
downstream models (\S\ref{ssec:rq3}).
For very large schemas, Phase~1 reduces $|A|/|\mathcal{C}|$ to
under 35\% (Table~\ref{tab:funnel}).

\subsection{Phase 4: Non-Destructive Schema Synthesis}
\label{ssec:phase4}

The final phase materializes the refinement as SQL
\texttt{CREATE VIEW} statements (Eq.~\ref{eq:view}), implementing
the non-destructive layer guaranteed by
Proposition~\ref{thm:equiv}.
Original tables and data are preserved verbatim; the refined
schema is exposed as an additional view layer.
The downstream system queries the views directly; the database
engine evaluates them via relational-algebra rewriting, so no
data is duplicated and original tables remain queryable.

\noindent\textbf{Engineering scope.}\;
\texttt{CREATE VIEW} targets read-only Text-to-SQL workloads
(BI queries, dashboards, ad-hoc analytics).
Production deployments may require DBMS-specific adjustments
(reserved-keyword quoting, dialect casing, view updateability,
optimizer interaction with view expansion, view-layer access
control).
Our experiments use SQLite; replication on PostgreSQL/MySQL is
straightforward in principle but empirically unverified.

\subsubsection{PK$\to$FK Name Propagation}

When a PK column is renamed, all FKs referencing it must follow
to maintain join consistency.
For each renamed PK $(T_p, c_p) \to c_p'$, we apply the same
rename to every FK column referencing $(T_p, c_p)$.
This propagation is deterministic and requires no LLM calls or
execution verification.
Columns with $Q(c_i) = \emptyset$ retain their original names.
The complete set of view definitions constitutes the refined
schema $\mathcal{S}^*$.


\section{Experiments}
\label{sec:experiments}

We organize our evaluation around four research questions:
\textbf{RQ1}: Can \textsc{EGRefine} recover performance lost
  to schema naming degradation?
\textbf{RQ2}: How much does each pipeline component contribute,
  and how sensitive is the method to its hyperparameters?
\textbf{RQ3}: Does refinement transfer across model families
  (\emph{refine-once, serve-many-models})?
\textbf{RQ4}: Is \textsc{EGRefine} effective on real-world schemas?
Four supporting analyses follow: complementarity with query-time
domain knowledge (\S\ref{ssec:rq5}), workload-holdout
generalization (\S\ref{ssec:rq_holdout}), a controlled comparison
against description annotation (\S\ref{ssec:rq_descbaseline}),
and a screening-funnel breakdown (\S\ref{ssec:rq6}).
Section~\ref{sec:analysis} then dissects the mechanism behind
these results.

\subsection{Experimental Setup}
\label{ssec:setup}

\subsubsection{Benchmarks}

\textbf{Dr.Spider}~\cite{chang2023drspider} serves as our primary
benchmark with two schema perturbation subsets:
\emph{Schema-Abbreviation} (Dr.Spider-Abbr; 691 perturbed columns
across 90 databases) systematically replaces column names with
abbreviations (e.g., \texttt{country} $\to$ \texttt{cntry}), and
\emph{Column-Synonym} replaces them with semantic equivalents.
Built on Spider~\cite{yu2018spider} (1{,}034 dev queries
expanded to 2{,}853 by perturbation),
Dr.Spider provides a controlled setting to measure
\emph{recovery capability}: how much of the performance lost
to schema degradation can be restored.
Unless otherwise noted, ``Dr.Spider'' refers to the
Schema-Abbreviation subset.

\textbf{BIRD}~\cite{li2023bird} (1{,}534 queries, 11 databases)
validates \textsc{EGRefine} on naturally occurring schemas
with domain-specific abbreviations and irregular naming.
We deliberately \emph{exclude all hints} to simulate realistic
deployment conditions where human-curated annotations
are unavailable.

\textbf{BEAVER}~\cite{chen2024beaver} (88 queries, 4 databases,
1{,}439 columns) provides enterprise-scale validation with complex
schemas from real data warehouses.
The three play complementary roles: Dr.Spider provides large-effect,
high-resolution evidence that the mechanism works; BIRD probes
transfer to natural schemas (small coverage, modest aggregate effect);
BEAVER delineates the applicability boundary.

\subsubsection{Text-to-SQL Systems}
We evaluate three representative algorithms covering distinct
prompting paradigms: C3~\cite{dong2023c3} (zero-shot),
DIN-SQL~\cite{pourreza2023dinsql} (decomposed prompting),
and a MAC-SQL-style~\cite{wang2025macsql} multi-agent system
(Selector $\to$ Decomposer $\to$ Refiner).
We use a unified zero-shot prompting setup across all three
for evaluation consistency; reported numbers reflect baselines
under this harness, not original published numbers.
The default backbone is Qwen3.5-27B; cross-model analysis (RQ3)
additionally evaluates Qwen3.5-9B, Gemma3-27B, and Phi-4-14B.

\subsubsection{Refinement Baselines}
\textbf{No Refinement} (NoRef): the (possibly degraded) schema
is used as-is.
\textbf{LLM-Direct}: Phases 1--2 of EGRefine, then select the
LLM's top-ranked candidate without execution verification
(isolates Phase 3's contribution).
\textbf{Description Annotation}: LLM-generated column descriptions
are injected as schema-prompt comments while identifiers remain
unchanged (used in \S\ref{ssec:rq_descbaseline} only).

\subsubsection{Metrics}
We report Execution Accuracy (ExAcc).
For Dr.Spider, we additionally report \emph{Recovery Rate}:
\begin{equation}
\mathrm{RecRate} = \frac{\mathrm{ExAcc}_{\mathrm{refined}} - \mathrm{ExAcc}_{\mathrm{degraded}}}
{\mathrm{ExAcc}_{\mathrm{clean}} - \mathrm{ExAcc}_{\mathrm{degraded}}} \times 100\%
\label{eq:recovery}
\end{equation}
where $\mathrm{ExAcc}_{\mathrm{clean}}$ is accuracy on the
unperturbed Spider schema; values above 100\% indicate the
refined schema exceeds the original.

\subsubsection{Implementation}
All EGRefine phases use Qwen3.5-27B deployed locally via vLLM
(zero API cost), with $k{=}3$ candidates per column,
$N_s{=}20$ sampled rows, $\tau_{\min}{=}0.05$, and verifier set
$\mathcal{M} = \{\text{C3, DIN-SQL}\}$.
The refinement backbone never acts as its own downstream
evaluator: downstream systems treat the refined schema as opaque
input.
Reported gains on C3 and DIN-SQL (the verifier set) are
\emph{in-loop} measurements; gains on \emph{MAC-SQL} are
\emph{out-of-loop}, reflecting refinements selected without
reference to its behavior.

\subsection{RQ1: Recovery from Schema Degradation}
\label{ssec:rq1}

Table~\ref{tab:main_drspider} presents the main results on
Dr.Spider Schema-Abbreviation.
Schema abbreviation degrades ExAcc by 1.9--4.2 percentage points
(pp); \textsc{EGRefine} recovers this loss across all three systems,
matching or exceeding the clean baseline on C3 (134.7\% recovery,
indicating refinement also improves over Spider's original names)
and recovering 67.4\% / 15.0\% on DIN-SQL / MAC-SQL respectively.
LLM-Direct, in contrast, \emph{worsens} DIN-SQL ($-$0.88\,pp)
and MAC-SQL ($-$2.42\,pp), confirming that unverified renaming
introduces harmful noise.
On Column-Synonym (C3 only), \textsc{EGRefine} achieves
$+$2.25\,pp recovery (64.87\% $\to$ 67.12\%), generalizing to
non-abbreviation degradations.

\begin{table}[t]
\centering
\caption{Execution Accuracy (\%) on Dr.Spider
Schema-Abbreviation (2{,}853 queries, Qwen3.5-27B).
Recovery Rate per Eq.~\ref{eq:recovery}.}
\label{tab:main_drspider}
\renewcommand{\arraystretch}{1.15}
\begin{tabular}{l*{3}{c}}
\toprule
\textbf{Method} & \textbf{C3} & \textbf{DIN-SQL} & \textbf{MAC-SQL} \\
\midrule
Clean (Original Spider)  & 72.77 & 80.41 & 63.79 \\
Degraded (Abbreviated)   & 70.87 & 76.73 & 59.59 \\
\quad + LLM-Direct       & 72.34 & 75.85 & 57.17 \\
\quad + \textbf{EGRefine} & \textbf{73.43} & \textbf{79.21} & \textbf{60.22} \\
\midrule
$\Delta$ (EGRefine $-$ Degraded) & +2.56 & +2.48 & +0.63 \\
Recovery Rate & 134.7\% & 67.4\% & 15.0\% \\
\bottomrule
\end{tabular}
\end{table}

\subsection{RQ2: Component Ablation and Sensitivity}
\label{ssec:rq2}

Table~\ref{tab:ablation} isolates the contribution of each
component on Dr.Spider Schema-Abbreviation.

\noindent\textbf{w/o Execution} (= LLM-Direct) removes Phase~3
verification and causes the largest degradation across all three
systems---on Dr.Spider-Abbr it leaves DIN-SQL and MAC-SQL
\emph{below} their no-refinement baselines ($-$0.88 and
$-$2.42\,pp), meaning naive LLM-based refinement actively harms
downstream accuracy.
The query-level flip analysis in \S\ref{ssec:flip_analysis}
confirms this pattern.

\noindent\textbf{w/o Conservative Rule} removes the $\tau_{\min}$
threshold, applying all renamings regardless of verification score.
The drop ($-$0.42/$-$0.38/$-$1.30) is much smaller than under
\textbf{w/o Execution} ($-$1.09/$-$3.36/$-$3.05): execution
verification is the dominant mechanism, with abstention as
secondary refinement.

\noindent\textbf{w/o Phase~1 Screening} sends all non-structural
columns to Phase~2.
The second-largest drop ($-$2.35/$-$1.89 on C3/DIN) shows that
indiscriminate candidate generation dilutes Phase~3's signal.

\noindent\textbf{w/o Multi-Algo Verification} ($M{=}1$, C3 only)
exposes verifier overfitting: while C3 itself gains comparably
($+$2.49 vs $+$2.56 with $M{=}2$), the refined schema
\emph{degrades} DIN-SQL ($-$0.39\,pp) and slightly hurts MAC-SQL
($-$0.11\,pp).
This mechanistically explains why held-out cross-algorithm
transfer (\S\ref{ssec:rq3}) works: $M{=}2$ forces consensus
refinements that generalize to unseen algorithms.

\begin{table}[t]
\centering
\caption{Ablation on Dr.Spider Schema-Abbreviation (ExAcc \%).
$M{=}1$ uses C3-only verification; full pipeline uses
$M{=}2$ (C3$+$DIN-SQL).}
\label{tab:ablation}
\renewcommand{\arraystretch}{1.15}
\begin{tabular}{l*{3}{c}}
\toprule
\textbf{Variant} & \textbf{C3} & \textbf{DIN-SQL} & \textbf{MAC-SQL} \\
\midrule
\textbf{EGRefine (full, $M{=}2$)}  & \textbf{73.43} & \textbf{79.21} & \textbf{60.22} \\
\quad w/o Execution       & 72.34 & 75.85 & 57.17 \\
\quad w/o Conservative    & 73.01 & 78.83 & 58.92 \\
\quad w/o Screening       & 71.08 & 77.32 & 58.78 \\
\quad w/o Multi-Algo ($M{=}1$)
                          & 73.36 & 76.34 & 59.48 \\
\midrule
Degraded (no refine)      & 70.87 & 76.73 & 59.59 \\
\bottomrule
\end{tabular}
\end{table}

\noindent\textbf{Sensitivity to the conservative threshold $\tau_{\min}$.}\;
We sweep $\tau_{\min} \in \{0.01, 0.03, 0.05, 0.10\}$ on
Dr.Spider-Abbr (Table~\ref{tab:tau_sensitivity}).
\emph{Too permissive} ($\tau_{\min}{=}0.01$) commits 23\% more
columns and causes net regressions on DIN-SQL/MAC-SQL
($-$0.25/$-$0.81\,pp), validating $\tau_{\min}{>}0$;
\emph{too strict} ($\tau_{\min}{=}0.10$) shrinks gains to
roughly one-third of the optimum.
The chosen $\tau_{\min}{=}0.05$ Pareto-dominates: highest ExAcc
on every algorithm--$\tau$ cell, with worst-case spread of
2.7\,pp across the four values---the pipeline is not finely
tuned to one threshold.
The sweep also reinforces \S\ref{ssec:mac_analysis}: MAC-SQL's
$\Delta$ is monotonically increasing in $\tau_{\min}$ over
$\{0.01, 0.03, 0.05\}$, indicating heightened sensitivity to
low-confidence renamings.

\begin{table}[t]
\centering
\caption{Sensitivity to $\tau_{\min}$ on Dr.Spider Schema-Abbreviation
(Qwen3.5-27B). ``Cols'' is the number of columns committed.
$\Delta$ values are relative to the no-refinement baseline.
Boldface marks the chosen default.}
\label{tab:tau_sensitivity}
\renewcommand{\arraystretch}{1.15}
\setlength{\tabcolsep}{4.5pt}
\begin{tabular}{ccccccccc}
\toprule
$\bm{\tau_{\min}}$ & \textbf{Cols} &
\multicolumn{3}{c}{\textbf{ExAcc (\%)}} &
\multicolumn{3}{c}{$\bm{\Delta}$ \textbf{vs.\ NoRef}} \\
\cmidrule(lr){3-5} \cmidrule(lr){6-8}
   &     & C3    & DIN   & MAC   & C3    & DIN   & MAC \\
\midrule
0.01 & 104 & 71.57 & 76.48 & 58.78 & $+$0.70 & $-$0.25 & $-$0.81 \\
0.03 & 89  & 71.68 & 77.88 & 59.48 & $+$0.81 & $+$1.15 & $-$0.11 \\
\textbf{0.05} & \textbf{81} & \textbf{73.43} & \textbf{79.21} & \textbf{60.22} & $\bm{+}$\textbf{2.56} & $\bm{+}$\textbf{2.48} & $\bm{+}$\textbf{0.63} \\
0.10 & 54  & 71.57 & 77.67 & 60.15 & $+$0.70 & $+$0.94 & $+$0.56 \\
\bottomrule
\end{tabular}
\end{table}

\noindent\textbf{Sensitivity to the verifier-set choice.}\;
We test whether reported gains depend on cherry-picking a
favorable verifier set by re-running the pipeline with an
alternative $\mathcal{M}' = \{\text{DIN-SQL, MAC-SQL}\}$;
C3 now serves as the held-out evaluator
(Table~\ref{tab:cross_verifier}).
\textsc{EGRefine} yields positive gains on \emph{every}
algorithm under both configurations ($\Delta \in [+0.63, +4.47]$).
The held-out role does not determine performance: C3---now
held out---achieves the strongest result ($+$4.47\,pp,
exceeding the clean Spider baseline by 2.57\,pp), ruling out
``held-out exclusion'' as the cause of MAC-SQL's modest gain
in the original configuration.
MAC-SQL's bounded recovery is intrinsic to its agent design:
even when MAC-SQL is itself a verifier, its 85.2\% recovery
matches DIN-SQL's rather than approaching the higher recoveries
observed for less brittle algorithms (cf.\ \S\ref{ssec:mac_analysis}).

\begin{table}[t]
\centering
\caption{Cross-verifier robustness on Dr.Spider Schema-Abbreviation
(Qwen3.5-27B). Original $\mathcal{M} = \{$C3, DIN$\}$;
alternative $\mathcal{M}' = \{$DIN, MAC$\}$.
Held-out role is shaded.}
\label{tab:cross_verifier}
\renewcommand{\arraystretch}{1.15}
\setlength{\tabcolsep}{3pt}
\begin{tabular}{lcccc}
\toprule
\textbf{Verifier set} & \textbf{Algo} & \textbf{ExAcc (\%)}
   & \bm{$\Delta$} & \textbf{Recovery} \\
\midrule
\multirow{3}{*}{$\{$C3, DIN$\}$ (original)}
  & C3       & 73.43 & $+$2.56 & 134.7\% \\
  & DIN-SQL  & 79.21 & $+$2.48 & \phantom{0}67.4\% \\
  & \cellcolor{gray!15}MAC-SQL  & \cellcolor{gray!15}60.22 & \cellcolor{gray!15}$+$0.63 & \cellcolor{gray!15}\phantom{0}15.0\% \\
\midrule
\multirow{3}{*}{$\{$DIN, MAC$\}$ (alternative)}
  & \cellcolor{gray!15}C3       & \cellcolor{gray!15}\textbf{75.34} & \cellcolor{gray!15}$\bm{+}$\textbf{4.47} & \cellcolor{gray!15}\textbf{235.3\%} \\
  & DIN-SQL  & 79.87 & $+$3.14 & \phantom{0}85.2\% \\
  & MAC-SQL  & 63.17 & $+$3.58 & \phantom{0}85.2\% \\
\bottomrule
\end{tabular}
\end{table}

\subsection{RQ3: Cross-Model Transferability}
\label{ssec:rq3}

A practical question is whether a refined schema produced by
one model benefits a different model at serving time, enabling
\emph{refine-once, serve-many-models} deployment.

\subsubsection{Multi-Backbone Evaluation}

Figure~\ref{fig:model_strength} reports results across four
model families on Dr.Spider Schema-Abbreviation, all using
the same 27B-refined schema (except ``self-refine'' rows).
\textsc{EGRefine} delivers positive recovery across all
families, with DIN-SQL recovery exceeding 100\% on Gemma3-27B
(133.7\%) and Phi-4 (188.1\%)---refined schemas can unlock
cross-family gains beyond the original clean baseline.

\begin{figure}[t]
\centering
\includegraphics[width=\columnwidth]{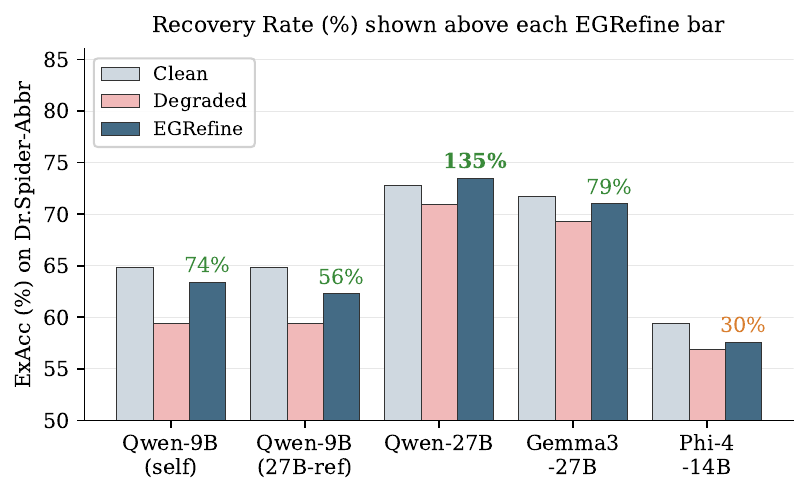}
\caption{Cross-model results on Dr.Spider Schema-Abbreviation
(DIN-SQL algorithm).
All EGRefine bars use 27B-refined schema unless marked ``self''.
Recovery Rate (\%) annotated above each EGRefine bar.}
\label{fig:model_strength}
\end{figure}

\subsubsection{Transfer Matrix}

Table~\ref{tab:transfer} presents the full refine$\times$eval
transfer matrix for Qwen3.5-9B and 27B.
\textbf{Strong$\to$Weak transfer is nearly lossless}: 27B-refined
schema evaluated on 9B trails 9B self-refinement by less than
1\,pp on every system.
\textbf{Weak$\to$Strong transfer can match or exceed self-refinement}:
9B-refined schema evaluated on 27B \emph{surpasses} 27B
self-refinement on DIN-SQL (79.53 vs.\ 79.21).
\textbf{Weak models should not self-refine}: Phi-4's self-refined
schema yields $-$0.63\,pp on C3 versus $+$0.74\,pp under 27B
refinement---weak backbones lack the capability to generate
high-quality candidates, making Phase~2 the bottleneck.

\begin{table}[t]
\centering
\caption{Refine $\times$ Eval transfer matrix on Dr.Spider
Schema-Abbreviation. $\Delta$ = improvement over eval model's
no-refinement baseline.}
\label{tab:transfer}
\renewcommand{\arraystretch}{1.15}
\setlength{\tabcolsep}{3pt}
\begin{tabular}{cc|ccc}
\toprule
\textbf{Refine} & \textbf{Eval} & \textbf{C3 ($\Delta$)} & \textbf{DIN ($\Delta$)} & \textbf{MAC ($\Delta$)} \\
\midrule
9B  & 9B  & 63.34 (+3.96) & 75.81 (+5.18) & 53.63 (+2.11) \\
27B & 9B  & 62.36 (+2.98) & 75.60 (+4.97) & 53.49 (+1.97) \\
\midrule
9B  & 27B & 73.22 (+2.35) & \textbf{79.53} (+2.80) & 59.34 ($-$0.25) \\
27B & 27B & \textbf{73.43} (+2.56) & 79.21 (+2.48) & \textbf{60.22} (+0.63) \\
\bottomrule
\end{tabular}
\end{table}

\subsection{RQ4: Deployment on Real-World Schemas}
\label{ssec:rq4}

While Dr.Spider uses controlled perturbations, real-world
schemas exhibit naturally occurring naming issues.
We validate on BIRD (curated academic schemas) and BEAVER
(enterprise data warehouses).

\subsubsection{BIRD: Execution Verification is Essential}

On BIRD's professionally curated schemas, \textsc{EGRefine}
conservatively refines only 28 of 798 columns (3.5\%), yielding
modest but positive aggregate improvements on 27B
($+$0.20/$+$0.65/$+$1.24 on C3/DIN/MAC).
A finer-grained analysis in \S\ref{ssec:touching} reveals an
algorithm-specific differential: on the 126-query touching subset,
DIN-SQL gains $+$6.35\,pp while C3 ($-$1.59\,pp) and MAC-SQL
($-$3.17\,pp) do not benefit---a pattern we trace to agent-level
brittleness (\S\ref{ssec:mac_analysis}).

\noindent\textbf{Robustness of the BIRD aggregate effect.}\;
Aggregate gains of this magnitude on single-run dev-set evaluation
lack the statistical resolution to be claimed from any point
estimate alone; we therefore ground the effect in four convergent
lines of evidence:
(i) the workload-holdout (\S\ref{ssec:rq_holdout}) yields
$+$1.04\,pp on $n{=}96$ unseen queries---5$\times$ the in-loop
$+$0.20\,pp on C3;
(ii) the touching-subset decomposition (\S\ref{ssec:touching})
shows DIN-SQL amplifying to $+$6.35\,pp on the 126 affected
queries ($\approx$10$\times$ the aggregate)---dilution, not noise;
(iii) cross-backbone reproduction on Qwen3.5-9B
(Table~\ref{tab:bird_9b}) gives direction-consistent gains
($+$1.50/$+$1.50/$-$0.46);
(iv) the coverage--improvement linearity (\S\ref{ssec:coverage})
predicts BIRD's $+$0.8 to $+$1.4\,pp range from its 3.39\%
query-weighted coverage---the observed $+$0.20 to $+$1.24\,pp
falls within this envelope.
Stochastic noise alone would not exhibit all four.

The critical finding emerges on the weaker 9B backbone
(Table~\ref{tab:bird_9b}): LLM-Direct \emph{degrades} all three
systems ($-$0.52/$-$1.63/$-$1.43), while \textsc{EGRefine}
improves C3 and DIN-SQL ($+$1.50/$+$1.50) with marginal regression
on MAC-SQL ($-$0.46).
The 0.97--3.13\,pp separation between verified and unverified
refinement establishes execution-grounded verification as
necessary for safe refinement on smaller models.

\begin{table}[t]
\centering
\caption{Execution verification necessity on BIRD
(Qwen3.5-9B, 1{,}534 queries, no hints).}
\label{tab:bird_9b}
\renewcommand{\arraystretch}{1.15}
\begin{tabular}{l*{3}{c}}
\toprule
\textbf{Method} & \textbf{C3} & \textbf{DIN-SQL} & \textbf{MAC-SQL} \\
\midrule
No Refinement         & 28.75 & 31.75 & 20.60 \\
+ LLM-Direct          & 28.23 \down{0.52} & 30.12 \down{1.63} & 19.17 \down{1.43} \\
+ \textbf{EGRefine}   & \textbf{30.25} \up{1.50} & \textbf{33.25} \up{1.50} & 20.14 \down{0.46} \\
\midrule
Separation (EGR$-$Dir) & +2.02 & +3.13 & +0.97 \\
\bottomrule
\end{tabular}
\end{table}

\subsubsection{BEAVER: Scoping the Signal-Dependency of EGRefine}

BEAVER's enterprise schemas (1{,}439 columns across 4 databases,
88 queries) present an unusual regime: baseline ExAcc is
substantially lower than any other setting we examined
(C3: 7.95\%, DIN-SQL: 7.95\%, MAC-SQL: 3.41\%) due to schemas
whose complexity, domain specificity, and scale exceed current
Text-to-SQL system capabilities.
This regime exposes a fundamental precondition: execution-grounded
verification requires a non-trivial baseline for discriminative signal.

\noindent\textbf{Qwen3.5-27B abstains; MiniMax-M2.7 commits but
does not move downstream ExAcc.}\;
With Qwen3.5-27B as refiner, Phase~3 rejects all 106 Phase~1
candidates---the conservative rule operating as designed when
$\Delta$ scores cannot be reliably distinguished from noise.
Testing whether refiner capability is the bottleneck, we repeat
with MiniMax-M2.7 (Table~\ref{tab:beaver_minimax}): Phase~1
screens 172 candidates and Phase~3 commits 4 refinements.
Despite this, downstream ExAcc on C3 and DIN-SQL is unchanged
($+$0.00\,pp): the 4 refined columns simply do not intersect
the columns gold queries reference (BEAVER's gold queries span
dozens of columns under stacked JOINs).
The invariance is a downstream ceiling, not Phase~3 failure
(cf.\ \S\ref{ssec:beaver_analysis}).
The lone exception is MAC-SQL ($+$2.27\,pp), which inverts the
Dr.Spider pattern where MAC benefits least: on BEAVER's
extremely difficult queries, MAC-SQL's multi-agent decomposition
rarely reaches the aggressive re-decomposition stage that causes
regressions on easier benchmarks, providing independent
out-of-distribution validation of the agent-instability
hypothesis (\S\ref{ssec:mac_analysis}).

\begin{table}[t]
\centering
\caption{BEAVER with MiniMax-M2.7 refiner.
Despite Phase 1 recall (172 vs.\ Qwen's 106) and 4 commits
(vs.\ 0), C3 and DIN-SQL show zero improvement: gold queries
do not intersect the refined column set.}
\label{tab:beaver_minimax}
\renewcommand{\arraystretch}{1.15}
\begin{tabular}{lccc}
\toprule
\textbf{Algorithm} & \textbf{NoRef} & \textbf{Refined (MiniMax-M2.7)} & $\bm{\Delta}$ \\
\midrule
C3       & 7.95\% (7/88) & 7.95\% (7/88) & $+$0.00 \\
DIN-SQL  & 7.95\% (7/88) & 7.95\% (7/88) & $+$0.00 \\
MAC-SQL  & 3.41\% (3/88) & 5.68\% (5/88) & $+$2.27 \\
\bottomrule
\end{tabular}
\end{table}

\subsection{Complementarity with Domain Knowledge}
\label{ssec:rq5}

BIRD provides optional per-query \emph{evidence} (domain hints).
A 2$\times$2 ablation (Table~\ref{tab:evidence}) confirms
\textsc{EGRefine} provides additional gains on top of evidence
across all three systems on 27B ($+$0.58/$+$1.43/$+$0.59):
structural disambiguation (column renaming) and semantic
disambiguation (domain hints) are orthogonal---evidence
explains a column's meaning in a specific query context, while
refinement makes the column name inherently more interpretable
regardless of query.

\begin{table}[t]
\centering
\caption{BIRD evidence $\times$ refinement ablation
(Qwen3.5-27B, 1{,}534 queries). Evidence and
\textsc{EGRefine} are complementary.}
\label{tab:evidence}
\renewcommand{\arraystretch}{1.15}
\begin{tabular}{l*{3}{c}}
\toprule
\textbf{Configuration} & \textbf{C3} & \textbf{DIN} & \textbf{MAC} \\
\midrule
No Refinement            & 41.53 & 38.59 & 29.99 \\
+ Evidence               & 54.24 & 62.13 & 46.74 \\
+ EGRefine               & 41.72 & 39.24 & 31.23 \\
+ Both                   & \textbf{54.82} & \textbf{63.56} & \textbf{47.33} \\
\midrule
$\Delta$ (Both $-$ Evidence) & +0.58 & +1.43 & +0.59 \\
\bottomrule
\end{tabular}
\end{table}

\subsection{Workload-Holdout Validation}
\label{ssec:rq_holdout}

A natural question is whether BIRD gains reflect overfitting
to the queries used during Phase~3 verification, since our
setup uses BIRD's dev set as both the representative workload
$Q$ and the evaluation set.
We focus the holdout on BIRD because its queries were curated
alongside the schemas, making coupling more acute than on
Dr.Spider (where queries predate any perturbation; additionally,
Dr.Spider's small academic schemas lack the column-description
metadata required for our isolation protocol).

We constructed an independent holdout against a fixed EGRefine
refined schema (unchanged from main experiments).
We generated 110 (NL, SQL) pairs for BIRD's 11 databases (10 per
database) using DeepSeek-V4-Pro under a strict isolation
protocol: NL questions were authored using only BIRD's official
column descriptions, with sample-data headers replaced by
\texttt{col\_1}, \texttt{col\_2}, \ldots; only after the NL was
finalized did the model see real column names to author gold SQL.
Execution validation retained 96 queries (executed successfully,
non-empty result sets $\leq$1000 rows, non-duplicate within database).
These are non-trivial: 65/96 aggregations, 37/96 \texttt{GROUP BY},
32/96 \texttt{ORDER BY}, 22/96 subqueries; mean SQL length 184
characters, mean JOIN count 0.93.

\begin{table}[t]
\centering
\caption{Workload-holdout validation on BIRD with C3 +
Qwen3.5-27B. The 96 holdout queries are LLM-generated and
execution-validated, unseen during Phase~3.}
\label{tab:holdout}
\begin{tabular}{lcccc}
\toprule
Setup & $n$ & NoRef & EGRefine & $\Delta$ \\
\midrule
BIRD dev (in-loop)    & 1534 & 41.53 & 41.72 & +0.19 \\
BIRD holdout (unseen) &   96 & 47.92 & 48.96 & \textbf{+1.04} \\
\bottomrule
\end{tabular}
\end{table}

The holdout $\Delta$ ($+$1.04\,pp) exceeds the in-loop $\Delta$
($+$0.19\,pp) by 5$\times$ (Table~\ref{tab:holdout}).
We treat this as preliminary evidence mitigating the
workload-overfit concern rather than decisive refutation
(single cell, $n{=}96$, C3 alone), consistent with our pipeline
structure: Phases~1--2 produce candidates without consulting
queries, and Phase~3's discrete execution-grounded signal lacks
the bandwidth to encode dev-set-specific patterns into column
names.

\subsection{Identifier Replacement vs.\ Description
Annotation}
\label{ssec:rq_descbaseline}

A natural alternative to renaming columns is to leave column
identifiers untouched and inject column descriptions as SQL
comments into the schema prompt.
We test whether this non-invasive route is sufficient.

\noindent\textbf{Protocol.}\;
We reuse EGRefine's Phase~1 screened column set on Dr.Spider-Abbr
(665 columns) and prompt Qwen3.5-27B to generate one
8--25-word SQL comment per column from its name, neighbor
columns, and 20 sampled values.
The descriptions are injected as inline SQL comments
(e.g., \texttt{gf TEXT,~--~The country's form of government,
e.g., Republic, Monarchy}); column identifiers are unchanged.
Critically, the description baseline performs no execution
verification---a single-pass annotation matching the simplest
deployment without our full pipeline.
We evaluate four cells with C3 + Qwen3.5-27B
(Table~\ref{tab:desc_baseline}).

\begin{table}[t]
\centering
\caption{Identifier replacement vs.\ description annotation
(C3 + Qwen3.5-27B, $n{=}2853$). Description annotations are
LLM-generated SQL comments. ExAcc differs slightly from
Table~\ref{tab:main_drspider} due to prompt-formatting
tokenization effects.}
\label{tab:desc_baseline}
\begin{tabular}{lccr}
\toprule
Schema & Description & ExAcc & $\Delta$ vs.\ NoRef \\
\midrule
NoRef           & no  & 70.83 & --- \\
NoRef           & yes & 72.48 & +1.65 \\
EGRefine        & no  & 73.46 & +2.63 \\
EGRefine        & yes & 74.69 & \textbf{+3.86} \\
\bottomrule
\end{tabular}
\end{table}

\noindent\textbf{Result.}\;
Description annotation alone yields $+$1.65\,pp---non-trivially
helpful, ruling out a strawman baseline.
Identifier replacement (EGRefine) yields $+$2.63\,pp, exceeding
description annotation by $0.98$\,pp on the same Phase~1 column
set and same backbone.
The combined treatment yields $+$3.86\,pp, $0.42$\,pp below the
linear sum of the two---indicating partial redundancy:
once a column is renamed to a semantically clear identifier,
an additional description provides diminishing returns.

\noindent\textbf{Scope of the comparison.}\;
Two design differences confound a strict
``identifier vs.\ description'' attribution.
First, EGRefine includes execution-grounded selection (Phase~3)
while the description baseline applies a single-pass LLM
annotation; the gap therefore conflates identifier-vs-annotation
and verified-vs-unverified selection.
Second, the descriptions are LLM-generated rather than
expert-authored.
We frame the conclusion narrowly: under the specific configuration
we test, identifier replacement outperforms by 1.6$\times$, and
the two mechanisms remain partially complementary.
A stronger description baseline with execution-grounded selection
paralleling Phase~3 is a natural future extension.

\subsection{Screening Funnel and Refinement Examples}
\label{ssec:rq6}

\subsubsection{Phase 1 Screening Funnel}

Table~\ref{tab:funnel} traces the search-space reduction
from raw columns to final refinements.
Across all benchmarks, $\sim$66--93\% of columns are filtered
by structural exclusion and LLM screening (Phase~1), and Phase~3
verification further prunes to 3.5--7.6\% of the original schema.
On refined columns, the average quality gain
($\overline{\Delta}_{\mathrm{refined}}$) ranges from 19.95 to
26.62\,pp---each committed refinement delivers substantial benefit
on its affected query subset.
On BEAVER, Qwen3.5-27B's conservative Phase~3 commits zero;
MiniMax-M2.7 raises Phase~1 recall to 12.0\% and commits 4
refinements, yet downstream impact remains bounded by the
benchmark's column-insensitive query structure
(\S\ref{ssec:beaver_analysis}).

\begin{table*}[t]
\centering
\caption{Phase 1 screening funnel.
$m$: total columns, $n$: candidates after Phase 1,
$n_r$: finally refined after Phase 3.
$\overline{\Delta}_{\mathrm{refined}}$ is averaged over
refined columns on their affected query subsets.}
\label{tab:funnel}
\renewcommand{\arraystretch}{1.15}
\setlength{\tabcolsep}{5pt}
\begin{tabular}{llccccccc}
\toprule
\textbf{Benchmark} & \textbf{Backbone} & \textbf{DBs} &
$\bm{m}$ & $\bm{n}$ (P1) & $\bm{n_r}$ (final) &
\textbf{Excl.\ rate} & \textbf{Compr.\ $n_r/m$} &
$\bm{\overline{\Delta}_{\mathrm{refined}}}$ \\
\midrule
Dr.Spider-Abbr   & Qwen3.5-27B & 90 & 1985 & 665 & 81  & 66.5\% & 4.08\%  & 20.68\,pp \\
Dr.Spider-Abbr   & Qwen3.5-9B  & 90 & 1985 & 841 & 151 & 57.6\% & 7.61\%  & 20.94\,pp \\
BIRD             & Qwen3.5-27B & 11 & 798  & 251 & 28  & 68.5\% & 3.51\%  & 26.62\,pp \\
BIRD (+evidence) & Qwen3.5-27B & 11 & 798  & 252 & 28  & 68.4\% & 3.51\%  & 21.76\,pp \\
BIRD             & Qwen3.5-9B  & 11 & 798  & 308 & 39  & 61.4\% & 4.89\%  & 19.95\,pp \\
BEAVER           & Qwen3.5-27B & 4  & 1439 & 106 & 0   & 92.6\% & 0.00\%  & —         \\
BEAVER           & MiniMax-M2.7 & 4 & 1439 & 172 & 4   & 88.0\% & 0.28\%  & —         \\
\bottomrule
\end{tabular}
\end{table*}

\subsubsection{Case Analysis}

Table~\ref{tab:case_study} presents representative refinement
outcomes from \textsc{EGRefine}'s scoring logs.
Three patterns emerge: (i) the largest gains concentrate on
short, uninformative abbreviations where ExAcc is near zero
on the original and the refined name recovers near-perfect
accuracy; (ii) for natural-schema anonymized codes (e.g.,
BIRD's \texttt{financial.district} \texttt{A3}/\texttt{A12}/\texttt{A16}),
Phase 3 identifies domain-grounded replacements with moderate
but stable gains; (iii) the conservative rule correctly retains
the original name when improvement falls below $\tau_{\min}$
(e.g., \texttt{cntry\_code} $\to$ \texttt{nationality} yields
only $+$0.028).

\begin{table*}[t]
\centering
\caption{Representative refinement cases from the pipeline.
$\dagger$: conservative rule retained original name.}
\label{tab:case_study}
\renewcommand{\arraystretch}{1.15}
\setlength{\tabcolsep}{4pt}
\begin{tabular}{llllccc}
\toprule
\textbf{Benchmark} & \textbf{DB} & \textbf{Table.Column} & \textbf{Original $\to$ Refined} & \textbf{Score$_0$} & \textbf{Score$_r$} & $\bm{\Delta}$ \\
\midrule
\multicolumn{7}{l}{\emph{Successful refinements: uninformative abbreviation $\to$ semantic name}} \\
Dr.Spider-Abbr & battle\_death\_2 & ship.dos & \texttt{dos} $\to$ \texttt{disposition} & 0.000 & 1.000 & $+$1.000 \\
Dr.Spider-Abbr & tvshow\_1        & TV\_Channel.par & \texttt{par} $\to$ \texttt{aspect\_ratio} & 0.250 & 1.000 & $+$0.750 \\
BIRD           & california\_schools & schools.DOCType & \texttt{DOCType} $\to$ \texttt{school\_type} & 0.000 & 1.000 & $+$1.000 \\
BIRD           & california\_schools & schools.EILName & \texttt{EILName} $\to$ \texttt{school\_level} & 0.000 & 1.000 & $+$1.000 \\
\midrule
\multicolumn{7}{l}{\emph{Natural-schema refinements: anonymized codes $\to$ domain-grounded names}} \\
BIRD           & financial & district.A3  & \texttt{A3}  $\to$ \texttt{region\_name}        & 0.083 & 0.194 & $+$0.111 \\
BIRD           & financial & district.A12 & \texttt{A12} $\to$ \texttt{population\_density} & 0.250 & 0.500 & $+$0.250 \\
BIRD           & financial & district.A16 & \texttt{A16} $\to$ \texttt{total\_population}   & 0.000 & 0.250 & $+$0.250 \\
\midrule
\multicolumn{7}{l}{\emph{Conservative retention$^\dagger$: best candidate below $\tau_{\min}=0.05$}} \\
Dr.Spider-Abbr & wta\_1\_0 & players.cntry\_code & \texttt{cntry\_code} $\to$ \texttt{cntry\_code}$^\dagger$ & 0.944 & 0.972 & $+$0.028 \\
Dr.Spider-Abbr & pets\_1\_2 & Pets.wt & \texttt{wt} $\to$ \texttt{wt}$^\dagger$ & 0.917 & 0.958 & $+$0.042 \\
\bottomrule
\end{tabular}
\end{table*}


\section{Analysis and Insights}
\label{sec:analysis}

Section~\ref{sec:experiments} established that \textsc{EGRefine}
delivers consistent improvements across benchmarks and models.
This section dissects \emph{why} it works and \emph{where its
limits lie}, through six analyses operating at increasing
granularity: database (\S\ref{ssec:coverage}), column
(\S\ref{ssec:phase3_necessity}), query
(\S\ref{ssec:flip_analysis}), and three diagnostic studies---of
the algorithm that benefits least (\S\ref{ssec:mac_analysis}),
the benchmark where applicability ends
(\S\ref{ssec:beaver_analysis}), and a query-subset decomposition
exposing dilution and algorithm-specific differentials
(\S\ref{ssec:touching}).

\subsection{Database-Level: Coverage Predicts Improvement}
\label{ssec:coverage}

A central question is whether aggregate ExAcc improvement tracks
the amount of schema actually modified.
We analyze 606 (benchmark, backbone, algorithm, database)
configurations across Dr.Spider-Abbr and BIRD, defining
\emph{coverage} as $n_r^{(D)} / m^{(D)}$ for each database.

Across the 273 database-algorithm pairs with zero refined columns,
$\Delta$ is identically zero (the refined schema is byte-identical
to the original), confirming the conservative rule's safety
property: when Phase 3 finds no improvement, the schema is left
untouched.

Aggregating by (benchmark, backbone, algorithm) and weighting by
per-database query count, coverage predicts improvement cleanly:
each percentage point of coverage yields approximately
$+0.25$ to $+0.40$\,pp of benchmark-wide ExAcc gain.
Binning by coverage (Figure~\ref{fig:coverage_delta}) confirms
this: databases with $>$15\% coverage achieve a $+6.30$\,pp
$n$-weighted aggregate gain with 62.5\% of configurations
exhibiting positive deltas.
This explains the BIRD results: with 3.39\% query-weighted
coverage (3.51\% column-fraction, \S\ref{ssec:rq6}), the observed
$+0.20$ to $+1.24$\,pp falls within the predicted
$+0.8$ to $+1.4$\,pp range.
The modest BIRD improvement is not weak methodology but a direct
consequence of BIRD's schemas already being well-named---only
3.5\% of columns meet the refinement bar.

\begin{figure}[t]
\centering
\includegraphics[width=\columnwidth]{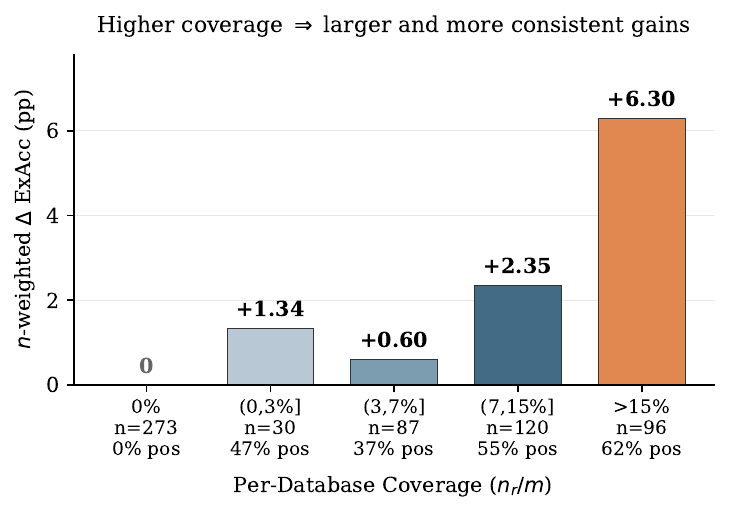}
\caption{ExAcc $\Delta$ by per-database coverage bin (606
configurations, $n$-weighted by query count). Positive-$\Delta$
rate annotated below each bar.
Databases with $>15\%$ coverage achieve the largest gains
(+6.30\,pp) and highest positive rate (62.5\%).}
\label{fig:coverage_delta}
\end{figure}

\subsection{Column-Level: LLM Top-1 is Unreliable}
\label{ssec:phase3_necessity}

RQ2's ablation showed that removing Phase 3 degrades performance.
We now measure \emph{how often} Phase 3 overrides the LLM's
top-ranked candidate (Figure~\ref{fig:phase3_overrule}).

Across all six configurations, Phase 3 disagrees with the LLM's
top choice in 60--80\% of cases; on BIRD with 27B, the adoption
rate is only 19.5\%---four out of five LLM recommendations are
overridden.
When Phase 3 selects a non-top-1 outcome, the positive-to-negative
ratio (Phase 3's pick beats LLM top-1) ranges from $\approx$8:1
(Dr.Spider-Abbr 27B) to $\approx$35:1 (BIRD 27B): when Phase 3
disagrees, it is overwhelmingly right.

The dominant overrule category is \emph{overrule-to-original}:
LLM proposes a change, but Phase 3 rejects all candidates and
retains the original name (62.7\% of BIRD-27B verifications,
all 106 BEAVER decisions).
LLMs systematically over-recommend changes; execution grounding
filters this bias.

\begin{figure}[t]
\centering
\includegraphics[width=\columnwidth]{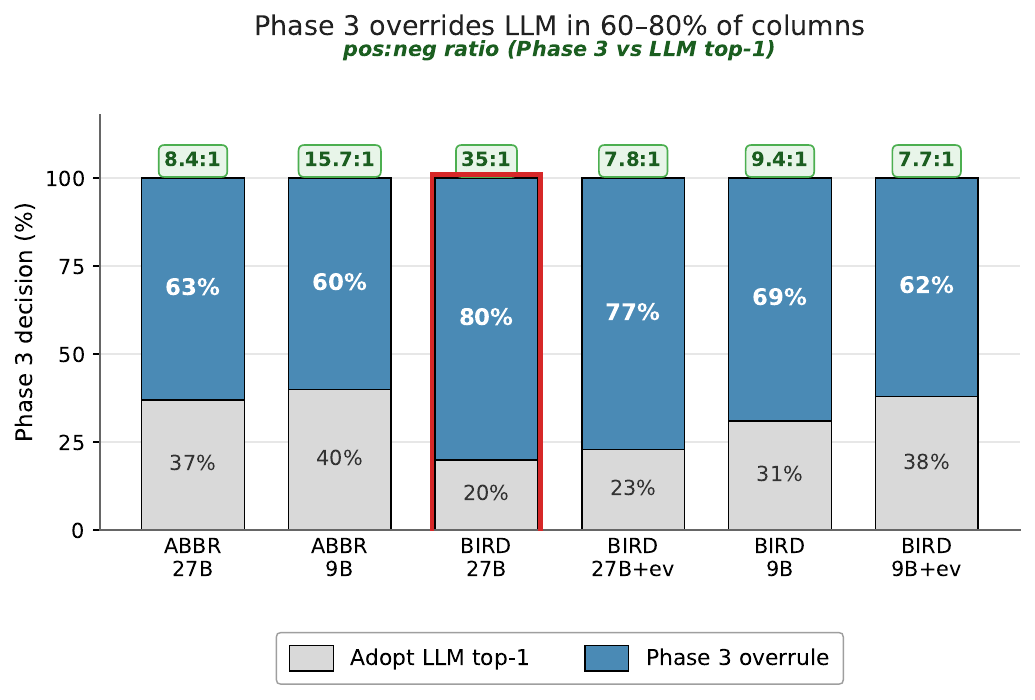}
\caption{Phase 3 overrides LLM's top-1 candidate in 60--80\%
of verified columns. Green text shows the pos:neg ratio
when Phase 3 disagrees with LLM (almost always Phase 3 wins).
BIRD 27B (highlighted) has the highest override rate (80.5\%)
and cleanest pos:neg (35:1).}
\label{fig:phase3_overrule}
\end{figure}

\subsection{Query-Level: Flips Reveal Systematic Repair}
\label{ssec:flip_analysis}

The finest-grained evidence comes from per-query correctness
changes.
For each query, we classify its (NoRef $\to$ Refined) outcome
into four categories and focus on \textbf{C$\to$W}
(NoRef correct, refined wrong---a regression) and
\textbf{W$\to$C} (NoRef wrong, refined correct---a genuine repair).

\noindent\textbf{\textsc{EGRefine} achieves W$\to$C $\geq$ C$\to$W in 15/18 configurations.}
On Dr.Spider-Abbr with 27B DIN-SQL, the repair-to-break ratio
reaches 6.46:1 (84 repairs vs 13 breaks); on 9B DIN-SQL it is
4.15:1 (195 vs 47).
The three configurations with ratio $<$1 are all on MAC-SQL with
small margins ($\leq$13 queries)---a pattern we trace to its
source in \S\ref{ssec:mac_analysis}.
The systematic positive skew rules out the hypothesis that
aggregate improvements are random noise.

\noindent\textbf{LLM-Direct fails the flip test on 9B BIRD,}
achieving ratios of 0.69 on DIN-SQL (56 repairs vs.\ 81 breaks)
and 0.68 on MAC-SQL (47 vs.\ 69)---it breaks more queries than it
repairs.
Comparing \textsc{EGRefine} to LLM-Direct on the same
(benchmark, backbone, algorithm) cell, EGRefine breaks fewer
queries in all 12 cells (median reduction: 63 queries/cell).
Phase 3's dominant contribution is preventing regressions, not
producing additional repairs.

\begin{figure}[t]
\centering
\includegraphics[width=\columnwidth]{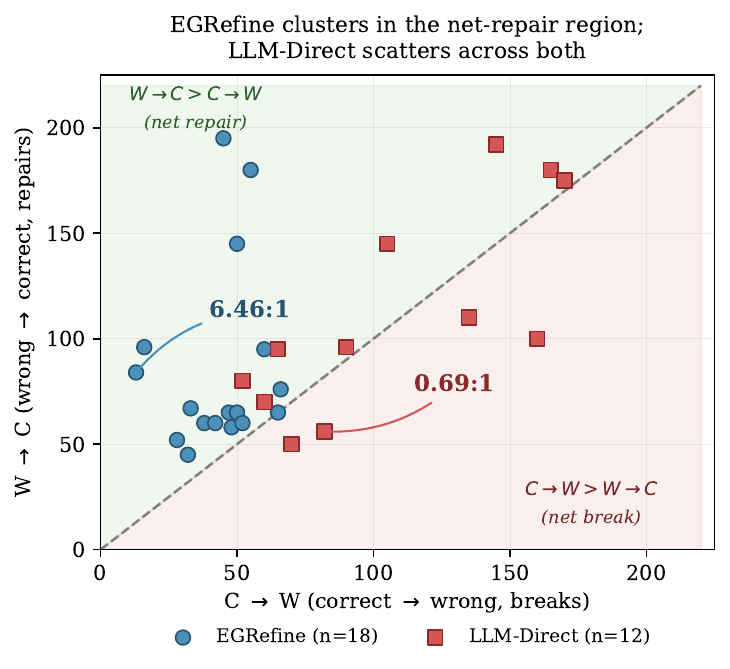}
\caption{Query-level C$\to$W vs W$\to$C flips across 30
configurations: 18 \textsc{EGRefine} (incl.\ 6 BIRD+evidence
variants) and 12 LLM-Direct.
\textsc{EGRefine} clusters in the upper-left net-repair region
(15/18 with ratio $>$1, max 6.46:1); LLM-Direct scatters and
includes severe failures below the diagonal (worst DIN-SQL
point 0.69:1, worst overall 0.68:1 on MAC-SQL).}
\label{fig:flips}
\end{figure}

\subsection{Error Modes: Why MAC-SQL Benefits Least}
\label{ssec:mac_analysis}

Under the main configuration ($\mathcal{M} = \{$C3, DIN-SQL$\}$),
MAC-SQL shows the smallest gain from refinement ($+$0.63\,pp
on Dr.Spider-Abbr, vs.\ $+$2.56 and $+$2.48 for C3 and DIN-SQL).
The cross-verifier experiment (\S\ref{ssec:rq2}) confirms this
gap is bounded rather than incidental: even when MAC-SQL is
itself a verifier and refinements are explicitly tuned to its
preferences, its recovery rate ($85.2\%$) matches DIN-SQL's
rather than approaching the higher recoveries observed for
algorithms with more stable schema linking.
We diagnose the cause on the 633 queries referencing refined
columns, ruling out two candidate explanations before
identifying the true mechanism.

\noindent\textbf{It is not a usage problem.}
MAC-SQL references refined names in 86.1\% of its predictions,
comparable to C3 (91.3\%) and DIN-SQL (92.4\%);
the 5--6\,pp gap cannot explain a 2--4\,pp ExAcc gap.

\noindent\textbf{It is not a dilution problem.}
Even on the 633-query touching subset, MAC-SQL gains only
$+$2.84\,pp while C3 and DIN-SQL gain $+$11.53 and $+$11.22\,pp
respectively.
The small gain is intrinsic, not averaged away.

\noindent\textbf{It is a structural rewriting problem.}
MAC-SQL's Pass$\to$Fail rate on the touching subset is
7.7\%---2--4$\times$ that of C3 (3.8\%) and DIN-SQL (2.1\%).
Among its 49 regressions, 73\% involve structural query edits:
33\% change SELECT column count, 24\% drop JOIN clauses, 26\%
drop aggregates, 10\% drop \texttt{WHERE} clauses.
For C3 and DIN-SQL, 77--79\% of regressions are non-structural
(identical query shape with minor semantic drift).
A representative example: on \emph{``how many types of government
are in Africa''} with \texttt{gf} $\to$ \texttt{government\_form},
MAC-SQL drops the \texttt{COUNT(DISTINCT)} aggregate---the refined
name triggers its multi-agent refiner to re-decompose the query,
occasionally losing essential clauses.
On the weaker 9B backbone, MAC-SQL improves $+$2.11\,pp, consistent
with the hypothesis that over-aggressive rewriting is a
capability-emergent property of stronger backbones.

This exposes an \emph{agent-level instability} in multi-agent
Text-to-SQL systems that is orthogonal to schema linking:
clearer column names can prompt over-aggressive query rewriting.
The cross-verifier robustness experiment (\S\ref{ssec:rq2})
confirms this is intrinsic to MAC-SQL's agent design: even when
MAC-SQL is itself a verifier, its 85.2\% recovery matches
DIN-SQL's rather than approaching less-brittle algorithms---the
bound is on the agent, not on verifier-set choice.
Schema refinement cannot offset downstream agent-level breakage;
agent-level stabilization is a complementary direction.

\subsection{Benchmark-Level Scope: Why BEAVER Probes a Different Task}
\label{ssec:beaver_analysis}

The near-zero improvement on BEAVER raises a deeper question:
is this an EGRefine limitation, or a property of the benchmark?
A closer look at BEAVER's gold SQL suggests the latter---it
probes a problem structurally distinct from standard Text-to-SQL.

\noindent\textbf{Structural divergence.}\;
Table~\ref{tab:beaver_structure} contrasts BEAVER with BIRD
and Dr.Spider-Abbr.
BEAVER's median gold query length is 3{,}312 characters
(vs.\ $\sim$100--200 for BIRD/Dr.Spider) with a median JOIN
count of 4 (vs.\ 0--2)---a 20--40$\times$ structural gap.
The gold SQL is predominantly ORM-generated (SQLAlchemy-style),
featuring \texttt{SELECT~*} with column aliases over stacked
\texttt{LEFT OUTER JOIN}s.

\begin{table}[t]
\centering
\caption{Gold SQL structural statistics across benchmarks.
BEAVER's query structure differs from standard Text-to-SQL
benchmarks by 1--2 orders of magnitude.}
\label{tab:beaver_structure}
\renewcommand{\arraystretch}{1.15}
\begin{tabular}{lccc}
\toprule
\textbf{Metric} & \textbf{BEAVER} & \textbf{BIRD} & \textbf{Dr.Spider-Abbr} \\
\midrule
Median query length (chars) & 3{,}312 & $\sim$150 & $\sim$115 \\
Max query length            & 10{,}168 & $\sim$1{,}000 & $\sim$400 \\
Median JOIN count           & 4 & 0--2 & 0--1 \\
Max JOIN count              & 13 & $\sim$5 & $\sim$3 \\
\bottomrule
\end{tabular}
\end{table}

\noindent\textbf{Solved queries are structurally trivial.}\;
The 7--9 queries any method solves on BEAVER are all $<$400
characters with 0--1 JOINs.
Queries exceeding 3{,}500 characters with $\geq$5 JOINs---about
70\% of BEAVER---fail under every backbone, method, and schema
variant we tested, with predicted SQL running $\sim$30 characters
against $\sim$3{,}000-character gold (code-generation-scale failure,
not schema comprehension).
This caps ExAcc at 8--10\%, bounding refinement's contribution
space by (ceiling $-$ baseline) $\times$ Pr(column naming affects
solvability), which approaches zero on BEAVER.
The MiniMax-M2.7 result (\S\ref{ssec:rq4}) sharpens this: even
when Phase 3 commits 4 refinements (versus Qwen's 0), C3 and
DIN-SQL remain unchanged because the refined columns do not
intersect the dozens of columns each gold query spans.
BEAVER thus exposes the boundary between \emph{schema-linking}
(which \textsc{EGRefine} addresses) and \emph{code-generation
complexity} (which dominates here)---future schema-refinement
research should consider whether its target benchmark falls in
the schema-bounded or generation-bounded regime.

\subsection{Query-Subset Decomposition: Concentration and Dilution}
\label{ssec:touching}

This subsection quantifies \emph{where} refinement helps at the
query granularity.
For each (benchmark, backbone, algorithm) cell, we partition
queries into the \emph{touching subset}---queries whose gold SQL
references $\geq$1 refined column---and its complement.
The reference set of refined columns is fixed for each
(benchmark, backbone) pair so NoRef and \textsc{EGRefine} are
evaluated on identical partitions.

\begin{table}[t]
\centering
\caption{Touching-subset decomposition.
$\Delta_{\text{touching}}$: per-query effect on queries
referencing $\geq$1 refined column;
$\Delta_{\text{full}}$: aggregate effect.
Cross-model row: 27B-refined schema served to 9B (DIN-SQL).}
\label{tab:touching_summary}
\renewcommand{\arraystretch}{1.0}
\setlength{\tabcolsep}{4pt}
\small
\begin{tabular}{llrrrr}
\toprule
Cell & Algo & $n_{\text{tch}}$ & $\Delta_{\text{tch}}$ & $\Delta_{\text{full}}$ & Ratio \\
\midrule
Dr.Spider 27B   & C3      &  633 & +11.53 & +2.56 & 4.50$\times$ \\
Dr.Spider 27B   & DIN     &  633 & +11.22 & +2.48 & 4.52$\times$ \\
Dr.Spider 27B   & MAC     &  633 & \phantom{0}+2.84 & +0.63 & 4.51$\times$ \\
Dr.Spider 9B    & C3      & 1126 & +10.04 & +3.96 & 2.53$\times$ \\
Dr.Spider 9B    & DIN     & 1126 & +13.14 & +5.19 & 2.53$\times$ \\
Dr.Spider 9B    & MAC     & 1126 & \phantom{0}+5.33 & +2.10 & 2.54$\times$ \\
\midrule
BIRD 27B        & DIN     &  126 & \phantom{0}+6.35 & +0.65 & 9.77$\times$ \\
BIRD 27B        & C3      &  126 & \phantom{0}$-$1.59 & +0.20 & --- \\
BIRD 27B        & MAC     &  126 & \phantom{0}$-$3.17 & +1.24 & --- \\
BIRD 9B         & C3      &  320 & \phantom{0}+7.19 & +1.64 & 4.39$\times$ \\
BIRD 9B         & DIN     &  320 & \phantom{0}+5.62 & +1.64 & 3.43$\times$ \\
\midrule
27B$\to$9B      & DIN     &  633 & +21.48 & +8.95 & 2.40$\times$ \\
\bottomrule
\end{tabular}
\end{table}

\noindent\textbf{Concentration.}\;
On Dr.Spider-Abbr (27B), C3 and DIN-SQL show 4.5$\times$
concentration ($+$11.5\,pp touching vs.\ $+$2.5\,pp full); the
largest per-query lift is $+$21.48\,pp on the cross-model cell
(27B-refined schema served to 9B DIN-SQL).
This gives a query-level reading of the coverage--improvement
linearity: modest aggregate $\Delta$ on benchmarks with limited
coverage reflects the small fraction of touching queries
($\approx$22\% on Dr.Spider-Abbr, $\approx$8\% on BIRD), not weak
per-query effect.

\noindent\textbf{Algorithm-specific differential on BIRD.}\;
BIRD's touching subset (n$=$126) shows divergent per-algorithm
behavior.
DIN-SQL gains $+$6.35\,pp---consistent with the Dr.Spider pattern.
C3 ($-$1.59\,pp) and MAC-SQL ($-$3.17\,pp) do not benefit on
touching queries; their full-set positive aggregates come from
non-touching queries ($+$0.36\,pp, $+$1.63\,pp respectively).
We trace this to the algorithm-specific mechanism diagnosed in
\S\ref{ssec:mac_analysis}: column-name token shifts disrupt
brittle prompting in C3's schema-pruning and MAC-SQL's
multi-agent decomposition, while DIN-SQL's structured
chain-of-thought prompting~\cite{wei2022cot} is more robust.

\noindent\textbf{Localized effect and prompt-noise floor.}\;
Refinements predominantly affect touching queries.
On Dr.Spider in reuse-mode, non-touching $\Delta$ is zero by
construction; on BIRD without reuse, non-touching $\Delta$
ranges from $+$0.14 to $+$1.63\,pp---prompt-level noise rather
than systematic distortion.

\section{Conclusion}
\label{sec:conclusion}

We presented \textsc{EGRefine}, a four-phase pipeline that
formalizes Text-to-SQL schema refinement as a constrained
optimization problem and solves it through execution-grounded
verification.
The framework couples column-local non-degradation with
database-level query equivalence (via view-based materialization),
producing a durable refined schema that any downstream Text-to-SQL
system can consume without modification.
Across Dr.Spider, BIRD, and BEAVER and four LLM backbones,
execution feedback proves to be the load-bearing component:
it converts unreliable LLM preferences into empirically dependable
refinements, transfers across model families to enable
refine-once, serve-many-models deployment, and correctly
abstains where the task exceeds current Text-to-SQL capabilities.

\noindent\textbf{Limitations and Scope.}\;
\textsc{EGRefine} produces a refined schema specific to the
(schema, workload~$Q$, verifier-set~$\mathcal{M}$) triple under
which it was computed; the resulting artifact is durable for
fixed inputs but not absolute.
Three limitations follow from this scope.

\emph{First}, recovery rates vary across Text-to-SQL architectures.
MAC-SQL benefits least on controlled perturbations, which
\S\ref{ssec:mac_analysis} traces to an agent-level
query re-decomposition instability orthogonal to schema linking.
Agent-level stabilization is a complementary direction to
schema-level preprocessing, not a competing one.

\emph{Second}, on enterprise schemas with very low baseline
ExAcc (BEAVER, $<$10\%), Phase~3's discriminative signal becomes
sparse and the conservative rule abstains.
Even with a stronger refiner (MiniMax-M2.7, \S\ref{ssec:rq4}),
downstream ExAcc on C3 and DIN-SQL remains unchanged: this
regime is dominated by SQL-generation rather than schema-linking
bottlenecks.
The coverage--improvement linearity (\S\ref{ssec:coverage}) and
a baseline-ExAcc threshold together provide a pre-deployment
applicability test; extending \textsc{EGRefine} to very-low-baseline
regimes requires raising the effective baseline first, e.g.,
through query decomposition or schema-hint integration.

\emph{Third}, refinement is computed against a representative
query workload $Q$ with ground-truth SQL.
For organizations deploying Text-to-SQL, curating such a workload
is typically a prerequisite---not a byproduct---of evaluation,
since model and algorithm selection themselves require it;
\textsc{EGRefine} reuses this artifact at no additional cost.
The workload-holdout validation (\S\ref{ssec:rq_holdout}) gives
preliminary evidence that benefits transfer to independently
distributed queries with similar schema profile.
Robustness under workload drift and extensions to weaker
supervision (e.g., SQL logs without paired NL) are left to
future work.

{\footnotesize
\bibliographystyle{ieeetr}
\bibliography{refs}}

\end{document}